\documentclass[a4paper]{article}

\usepackage[utf8]{inputenc}

\usepackage[dvipdfmx]{graphicx}
\usepackage{bmpsize}
\usepackage{multicol}
\usepackage{amsthm}
\usepackage{bm}
\usepackage{amssymb}
\usepackage{amsmath}
\usepackage{color}
\usepackage{ulem}
\usepackage{amsthm}
\usepackage{subcaption} 
\usepackage{graphicx}
\usepackage{dcolumn}
\usepackage{bm}

\usepackage[utf8]{inputenc}
\usepackage[T1]{fontenc}
\usepackage{etoolbox}
\usepackage{ascmac,wrapfig,makeidx}
\usepackage{physics}
\usepackage[stable]{footmisc}
\usepackage{mathrsfs}
\usepackage{color}
\usepackage{comment}
\usepackage{amsmath}

\theoremstyle{plain}

\newtheorem*{theorem*}{Theorem}
\newtheorem{lemma}{Lemma}[section]

\theoremstyle{definition}

\newcommand{\BA}{\begin{eqnarray}}
\newcommand{\EA}{\end{eqnarray}}

\definecolor{dgreen}{rgb}{0.0, 0.5, 0.0}

\begin{document}

\fontsize{14pt}{16.5pt}\selectfont

\begin{center}
\bf{General Construction of Bra-Ket Formalism for Identical
Particle Systems in Rigged Hilbert Space Approach
}
\end{center}
\fontsize{12pt}{11pt}\selectfont
\begin{center}
S. Ohmori$^{1,*}$ and J. Takahashi$^2$\\ 
\end{center}

\bigskip

\noindent
\it{1)~Department of Economics, Hosei University, Machida-shi, Tokyo 194-0298,
Japan.}
\noindent
\it{2)~Department of Economics, Asia University, Musashino-shi, Tokyo 180-0022,
Japan.}

\bigskip

\noindent
*corresponding author: 42261timemachine@ruri.waseda.jp\\
~~\\
\rm
\fontsize{11pt}{14pt}\selectfont\noindent

\baselineskip 20pt

{\bf Abstract}\\
%
This study discussed Dirac's bra-ket formalism for the identical particles system based on the rigged Hilbert space reformulated by R.~Madrid [J. Phys A:Math. Gen. 37, 8129 (2004)]. 
The bra and ket vectors for a composite system that form the basis of an identical particle system are described in dual and anti-dual spaces for the tensor product of rigged Hilbert spaces.
The permutation operator that characterizes the symmetry of identical particles is constructed as the operator on such dual spaces.
We also show that the nuclear spectral theorem in the tensor product of rigged Hilbert spaces endows the spectral expansion of the self-adjoint operator in the dual and anti-dual spaces and the expansion is consistent with the identicle particle system when the permutation operator commutes the self-adjoint operator.
%
%

%


\bigskip

\section{Introduction}
\label{sec:1}

A mathematical approach utilizing rigged Hilbert space (RHS) has been developed to handle Dirac's bra-ket notations precisely~\cite{Robert1966a,Robert1966b,Antoine1969a,Antoine1969b,Melsheimer1974a,Melsheimer1974b,Bohm1978,Bohm1981,Prigogine1996,Bohm1998,Antoiou1998,Antoiou2003,Gadella2003,Madrid2004,Madrid2005,Antoine2009,Antoine2021}.
%
%
%
RHS comprises the following triplet of topological vectors spaces~\cite{Gelfand1964,Maurin1968}, 
\begin{equation}
    \Phi \subset \mathcal{H} \subset \Phi^\prime,
    \label{eqn:00}
\end{equation}
where $\mathcal{H}=(\mathcal{H}, \langle \cdot, \cdot \rangle_\mathcal{H})$ is a complex Hilbert space and $\Phi=(\Phi, \tau_\Phi)$ is a nuclear space that is a dense linear subspace of $\mathcal{H}$.
The inner product $\langle \cdot, \cdot \rangle_\Phi$ on $\Phi$ becomes separately continuous on $(\Phi, \tau_\Phi)$, where $\langle \phi, \psi \rangle_\Phi \equiv \langle \phi, \psi \rangle_\mathcal{H}$ for $\phi, \psi \in \Phi$.
$\Phi^\prime$ is a family of continuous linear functionals on $(\Phi,\tau_\Phi)$.
In the case of the RHS approach, the nuclear spectral theorem for a self-adjoint operator (observable) in $\mathcal{H}$ guarantees the existence of generalized eigenvectors that characterize the eigenequations for the bra and ket vectors, individually.
This theorem also provides the spectral expansions based on which the spectral decomposition for discrete and continuous spectrum, specified by
Dirac's $\delta$-function (distributions) found in the literature, can be constructed.
Hence, Dirac's bra-ket formalism is subsumed within RHS framework, which is considered as the foundational framework of quantum mechanics.
Indeed, several studies using the RHS approach have developed precise and elegant formulations to address various problems in quantum theory, such as the harmonic oscillator~\cite{Bohm1978}, resonance states (Gamow vectors)~\cite{Bohm1981}, and scattering problems~\cite{Madrid2004}.

Recently, this approach has begun to be applied to modern quantum physics, such as resonance states in open quantum systems and non-Hermitian operators exhibiting characteristic symmetries~\cite{Chruscinski2003,Chruscinski2004,Lars2019,Fernandez2022,Ohmori2022,Ohmori2024}. 
Note that the physical phenomena observed in these systems cannot be adequately described using only the Hilbert space, such as the $L^2$-space. 
For instance, in the problem of a quantum damped system, if the $L^2$-space is treated as the fundamental space, the given Hamiltonian shows only real spectrum, whereas, it contains complex eigenvalues when RHS is selected.
Then, the complex eigenvalues can be interpreted as the resonant state~\cite{Chruscinski2003,Chruscinski2004}. 
Thus, the RHS is indispensable for addressing complex eigenvalues beyond the $L^2$-space theory.
As evident from this example, we believe that the development of an RHS theory is crucial for the mathematical foundations and the elucidation of the quantum phenomena.

To construct the bra and ket vectors using RHS, 
a more elegant and simple approach, proposed by Madrid~\cite{Madrid2004}, has been developed.
This approach adapts the RHS (\ref{eqn:00}) including the dual space $\Phi^\times$ of $\Phi$, 
\begin{equation}
    \Phi \subset \mathcal{H} \subset \Phi^\prime,\Phi^\times,
    \label{eqn:O1-1}
\end{equation}
where $\Phi^\times$ is a family of continuous {\it anti-linear} functionals on $(\Phi,\tau_\Phi)$.
(A function $f\in \Phi^{\times}$ is anti-linear if it satisfies $f(a\varphi+b\phi)=a^*f(\varphi)+b^*f(\phi)$ where $a$ and $b$ are complex numbers with complex conjugates $a^*$ and $b^*$ and $\varphi, \phi \in \Phi$.)
Using (\ref{eqn:O1-1}), the bra and ket vectors are established as elements of $\Phi^\prime$ and $\Phi^\times$, in the following procedure. 
Let $\varphi \in \Phi$, and we define a map
$\ket{\varphi}_\mathcal{H} : \Phi \rightarrow \mathbb{C}^1$ using
$\ket{\varphi}_\mathcal{H}(\phi) \equiv  \langle \phi,\, \varphi\rangle_\mathcal{H}$ for $\phi \in \Phi$; this map is called a ket of $\varphi$.
The bra vector of $\varphi$ is defined as the complex conjugate of $\ket{\varphi}_\mathcal{H}$, 
namely, the map $\bra{\varphi}_\mathcal{H} : \Phi \rightarrow \mathbb{C}^1 $ where $\bra{\varphi}_{\mathcal{H}}(\phi)= \ket{\varphi}_{\mathcal{H}}^*(\phi)=(\ket{\varphi}_{\mathcal{H}}(\phi))^*=\langle \varphi,\, \phi\rangle_\mathcal{H}$.
%
%
Clearly, $\bra{\varphi}_{\mathcal{H}}$ and $\ket{\varphi}_{\mathcal{H}}$ belong to $\Phi^{\prime}$ and $\Phi^{\times}$ of $\Phi$, respectively.
The combination of dual and anti-dual spaces, $\Phi^\prime$ and $\Phi^\times$, is hereafter referred to as the dual spaces.
%
%
%
In the description, the spectral expansions of the self-adjoint operator can be performed as the elements of the dual spaces, in which all calculations in terms of the bra and ket vectors are conducted.
%
This Madrid approach exactly supplies the rigorous formalism of bra-ket notation.
However, it remains insufficient for composite systems containing identical particle systems compared to single-particle systems.
%
%
%
This study aimed to construct the bra-ket space in the dual spaces for identical particles based on Madrid's RHS formalism.

The remainder of this paper is organized as follows.
In Section~\ref{sec:2}, 
we construct the bra-ket vectors for the tensor product of the RHS (\ref{eqn:O1-1}) on the dual spaces and show the relation to the single bra-ket vectors obtained from an RHS.
In addition, the permutation operator is introduced on the dual spaces, which endows the symmetric properties of the bra-ket vectors derived from the identical RHS.
Using the nuclear spectral theorem for the tensor product of the RHS,
we present the formulation of the spectral expansions of the bra-ket vectors by the generalized eigenvectors for a self-adjoint operator in {Secion~\ref{sec:3}}. 
The generalized eigenvectors form a complete orthonormal system in the dual spaces.
Furthermore, the permutation operator obtained in Section \ref{sec:4} aided in the generalization of the eigenvectors, thus preserving the symmetric structure.
%
%
%
Finally, Section~\ref{sec:5} presents the conclusions.

\section{Construction of the bra-ket vectors in the dual spaces}
\label{sec:2}

\subsection{General formulation}
\label{sec:2.1}

When establishing the state space that describes a composite system without interactions using Hilbert space theory, the tensor product of Hilbert spaces is introduced~\cite{Simon1980,Hall2013}.
Similarly, in the RHS context, 
the tensor product of RHS is required to construct the bra and ket vectors related to a composite system.
For simplicity, we focus on a two-particle system.
Let $\Phi_i \subset \mathcal{H}_i \subset \Phi _i^\prime, \Phi _i^\times $ ($i=1,2$) be a RHS (\ref{eqn:O1-1}), where each $(\mathcal{H}_i,\langle \cdot, \cdot \rangle_ {\mathcal{H}_i})$ is a complex Hilbert space, $\Phi_i=(\Phi_i,\tau_{\Phi_i})$ is a subspace of $\mathcal{H}_i$ with the nuclear topology $\tau_{\Phi_i}$, and $\Phi _i^\prime$ and $\Phi _i^\times$ are the dual and anti-dual spaces of $(\Phi_i,\tau_{\Phi_i})$, respectively.
From each RHS, the bra and ket vectors are expressed as the maps $\bra{\varphi}_{\mathcal{H}_i}$ and $\ket{\varphi}_{\mathcal{H}_i}$ 
in $\Phi _i^\times$ and $\Phi _i^\prime$, respectively ($i=1,2$).
Now we introduce the algebraic tensor product for the Hilbert spaces $\mathcal{H}_1$ and $\mathcal{H}_2$ as an inner product space $\mathcal{H}_1 \otimes \mathcal{H}_2=(\mathcal{H}_1 \otimes \mathcal{H}_2,\langle \cdot,\, \cdot\rangle_{\mathcal{H}_1 \otimes \mathcal{H}_2})$ where $\mathcal{H}_1 \otimes \mathcal{H}_2=\Big{\{}\displaystyle\sum_{j=1}^m\varphi_{1j}\otimes\varphi_{2j}\mid \varphi_{1j}\in \mathcal{H}_1,\varphi_{2j}\in \mathcal{H}_2,j=1\sim m,m\in \mathbb{N} \Big{\}}$. 
Its inner product satisfies $\langle \varphi_1 \otimes \varphi_2,\, \phi_1\otimes\phi_2\rangle_{\mathcal{H}_1 \otimes \mathcal{H}_2}=\langle \varphi_1 \, \phi_1 \rangle_{\mathcal{H}_1}\langle \varphi_2 \, \phi_2 \rangle_{\mathcal{H}_2}$.
The completion of the algebraic tensor product with respect to the topology induced by $\langle \cdot,\, \cdot\rangle_{\mathcal{H}_1 \otimes \mathcal{H}_2}$ is denoted by $\mathcal{H}_1 \bar {\otimes} \mathcal{H}_2=(\mathcal{H}_1 \bar{\otimes} \mathcal{H}_2,\langle \cdot,\, \cdot\rangle_{\mathcal{H}_1 \bar{\otimes} \mathcal{H}_2})$.
The algebraic tensor product of the nuclear spaces $(\Phi_1,\tau_{\Phi_1})$ and $(\Phi_2,\tau_{\Phi_2})$ is also expressed as a locally convex space $\Phi_1 \otimes {\Phi}_2=\Big{\{}\displaystyle\sum_{j=1}^m\varphi_{1j}\otimes\varphi_{2j}\mid \varphi_{1j}\in {\Phi}_1,\varphi_{2j}\in {\Phi}_2,j=1\sim m,m\in \mathbb{N} \Big{\}}$ equipping the locally convex topology $\tau_p$ with the local base $\mathcal{B}_p=\{\Gamma(V_1\otimes V_2) \mid V_i\in \mathcal{B}_i,i=1,2\}$ where each $\mathcal{B}_i$ is a local base of $\tau_{\Phi_i}$ and $\Gamma(X)$ stands for the convex circled hull of a set $X$~\cite{Schaefer1966}.
As is well-known, the completion of $(\Phi_1 \otimes {\Phi}_2,\tau_p)$ is the nuclear space, 
denoted by $(\Phi_1 \hat{\otimes} {\Phi}_2,\widehat{\tau_p})$.
Therefore, considering the dual and anti-dual spaces of $\Phi_1 \hat{\otimes} {\Phi}_2$,
it is verified that the following triplet comprises an RHS~\cite{Maurin1968},
\begin{eqnarray}
    \Phi_1 \hat{\otimes} {\Phi}_2 \subset \mathcal{H}_1 \bar {\otimes} \mathcal{H}_2 \subset 
    (\Phi_1 \hat{\otimes} {\Phi}_2)^{\prime},~(\Phi_1 \hat{\otimes} {\Phi}_2)^{\times}.
    \label{eqn:O2-4}
\end{eqnarray}
Using the RHS (\ref{eqn:O2-4}) the bra and ket vectors corresponding to $\varphi\in \Phi_1 \hat{\otimes} {\Phi}_2$ are defined by
\begin{eqnarray}
    \bra{\varphi}_{\mathcal{H}_1 \bar{\otimes} \mathcal{H}_2} : \Phi_1 \hat{\otimes} {\Phi}_2 \to \mathbb{C},~\bra{\varphi}_{\mathcal{H}_1 \bar{\otimes} \mathcal{H}_2}(\phi)=\langle \varphi,\, \phi\rangle_{\mathcal{H}_1 \bar {\otimes} \mathcal{H}_2}, \\
    \label{eqn:o2-5a}
    \ket{\varphi}_{\mathcal{H}_1 \bar{\otimes} \mathcal{H}_2} : \Phi_1 \hat{\otimes} {\Phi}_2 \to \mathbb{C},~\ket{\varphi}_{\mathcal{H}_1 \bar{\otimes} \mathcal{H}_2}(\phi)=\langle \phi,\, \varphi\rangle_{\mathcal{H}_1 \bar {\otimes} \mathcal{H}_2}.    
    \label{eqn:o2-5b}
\end{eqnarray}
Consequently, the relations $\ket{\varphi}_{\mathcal{H}_1 \bar{\otimes} \mathcal{H}_2}=\bra{\varphi}_{\mathcal{H}_1 \bar{\otimes} \mathcal{H}_2}^*$, $\bra{\varphi}_{\mathcal{H}_1 \bar{\otimes} \mathcal{H}_2} \in (\Phi_1 \hat{\otimes} {\Phi}_2)^{\prime}$, and $\ket{\varphi}_{\mathcal{H}_1 \bar{\otimes} \mathcal{H}_2} \in (\Phi_1 \hat{\otimes} {\Phi}_2)^{\times}$, are satisfied.
To observe a connection between the ket $\ket{\varphi}_{\mathcal{H}_1 \bar{\otimes} \mathcal{H}_2}$ and the kets $\ket{\varphi}_{\mathcal{H}_i}~(i=1,2)$ for the single RHS, 
we consider $\varphi=\varphi_1 \otimes \varphi_2 \in \Phi_1 \otimes \Phi_2 \subset \Phi_1 \hat{\otimes} \Phi_2$.
The ket $\ket{\varphi}_{\mathcal{H}_1 \bar{\otimes} \mathcal{H}_2}$ becomes $\ket{\varphi}_{\mathcal{H}_1 \bar{\otimes} \mathcal{H}_2}=\ket{\varphi_1 \otimes \varphi_2}_{\mathcal{H}_1 \bar{\otimes} \mathcal{H}_2}$ in $(\Phi_1 \hat{\otimes} {\Phi}_2)^{\times}$.
By introducing a map $\ket{\varphi_1}_{\mathcal{H}_1}\ket{\varphi_2}_{\mathcal{H}_2} : \Phi_1 \times \Phi_2\to \mathbb{C}$ where $\ket{\varphi_1}_{\mathcal{H}_1}\ket{\varphi_2}_{\mathcal{H}_2}(\phi_1,\phi_2)=\langle \phi_1, \varphi _1 \rangle_ {\mathcal{H}_1}\langle \phi _2, \varphi _2 \rangle_ {\mathcal{H}_2}$ for $(\phi_1,\phi_2)\in \Phi_1 \times \Phi _2$, we obtain
\begin{eqnarray}
    \ket{\varphi_1 \otimes \varphi_2}_{\mathcal{H}_1 \bar{\otimes} \mathcal{H}_2}(\phi)
    =
    \langle \phi_1, \varphi _1 \rangle_ {\mathcal{H}_1}\langle \phi _2, \varphi _2 \rangle_ {\mathcal{H}_2}
    =\ket{\varphi}_{\mathcal{H}_1}\ket{\varphi}_{\mathcal{H}_2}(\phi_1,\phi_2),    
\label{relation1}
\end{eqnarray}
for $\phi=\phi_1 \otimes \phi_2 \in \Phi_1 \otimes \Phi_2$.
As $\ket{\varphi_1}_{\mathcal{H}_1}\ket{\varphi_2}_{\mathcal{H}_2}$ is anti-linear continuous on $\Phi_1 \times \Phi _2$, there exists the unique element $v$ of $(\Phi_1 \otimes \Phi _2)^{\times}$
satisfying $\ket{\varphi_1}_{\mathcal{H}_1}\ket{\varphi_2}_{\mathcal{H}_2}=v\circ \chi$,
namely, 
$v\circ \chi (\phi_1, \phi _2)=v(\phi_1 \otimes \phi _2)=\ket{\varphi_1}_{\mathcal{H}_1}\ket{\varphi_2}_{\mathcal{H}_2}(\phi_1,\phi_2)$ for any $(\phi_1, \phi_2)\in \Phi_1 \times \Phi _2$,
where $\chi : (\phi_1,\phi_2)\mapsto \phi_1\otimes\phi_2$ is the canonical bilinear map on $\Phi_1 \times \Phi_2$ into $\Phi_1 \otimes \Phi_2$~\cite{Schaefer1966}.
Note that the mapping $H : v \mapsto v\circ \chi$ becomes an isomorphism between $(\Phi_1 \otimes \Phi _2)^{\times}$ and $\mathcal{B}^{\times}(\Phi_1,\Phi_2)$ where
$\mathcal{B}^{\times}(\Phi_1,\Phi_2)$ is the family of continuous antilinear functionals on $(\Phi_1 \times \Phi _2, \tau_{\Phi_1 \times \Phi _2})$.
From (\ref{relation1}), the uniqueness of $v$ shows $v=\ket{\varphi_1\otimes \varphi_2}_{\mathcal{H}_1 \bar{\otimes} \mathcal{H}_2}|_{(\Phi_1 \otimes \Phi_2)}$.
($f|_A$ denotes the restriction of the map $f$ on $A$.)
In setting $\ket{\varphi_1\otimes \varphi_2}_{\mathcal{H}_1 \bar{\otimes} \mathcal{H}_2}|_{(\Phi_1 \otimes \Phi_2)}$ $\equiv$ $\ket{\varphi_1\otimes \varphi_2}_{\mathcal{H}_1 {\otimes} \mathcal{H}_2}$,
the isomorphism $H$ identifies 
\begin{eqnarray}
    \ket{\varphi_1 \otimes \varphi_2}_{\mathcal{H}_1 {\otimes} \mathcal{H}_2}
    =\ket{\varphi_1}_{\mathcal{H}_1}\ket{\varphi_2}_{\mathcal{H}_2}.   
\label{relation2}
\end{eqnarray}
Here, we set an isomorphic mapping $\hat{L} : \Phi _1^{\times}\otimes \Phi _2^{\times} \to \hat{L}(\Phi _1^{\times}\otimes \Phi _2^{\times}) \subset \mathcal{B}^{\times}(\Phi_1,\Phi_2)$ where $\hat{L}(f \otimes g)(\varphi,\phi)=f(\varphi)g(\phi)$ for $f\otimes g\in \Phi _1^{\times}\otimes \Phi _2^{\times}$ and $ (\varphi,\phi)\in \Phi_1 \times \Phi _2$.
By replacing $f$ and $g$ with the kets $\ket{\varphi_1}_{\mathcal{H}_1}$ and $\ket{\varphi_2}_{\mathcal{H}_2}$  in $\Phi_1^{\times}$ and $\Phi_2^{\times}$, respectively, we have $\hat{L}(\ket{\varphi_1}_{\mathcal{H}_1} \otimes \ket{\varphi_2}_{\mathcal{H}_2})=\ket{\varphi_1}_{\mathcal{H}_1}\ket{\varphi_2}_{\mathcal{H}_2}$.
By considering the isomorphism $\hat 
{L}$ as an identification, we obtain  
\begin{eqnarray}
    \ket{\varphi_1}_{\mathcal{H}_1} \otimes \ket{\varphi_2}_{\mathcal{H}_2}=\ket{\varphi_1}_{\mathcal{H}_1}\ket{\varphi_2}_{\mathcal{H}_2}. 
\label{relation3}
\end{eqnarray}
Thus, using (\ref{relation2}) and (\ref{relation3}), we obtain
\begin{eqnarray}
    \ket{\varphi_1 \otimes \varphi_2}_{\mathcal{H}_1 \bar{\otimes} \mathcal{H}_2}=\ket{\varphi_1}_{\mathcal{H}_1} \otimes \ket{\varphi_2}_{\mathcal{H}_2}
\label{eqn:ket_relation}
\end{eqnarray}
for any $\varphi_1 \in \Phi_1$ and $\varphi_2 \in \Phi_2$. 
(\ref{eqn:ket_relation})
indicates that the ket for $\varphi=\varphi_1\otimes \varphi_2$ defined by (\ref{eqn:o2-5b}) under the tensor product of RHS can be represented by the tensor product of the kets each of which is defined in a single RHS.
This connection (\ref{eqn:ket_relation}) is consistent with the assumption found in the physical literature; the ket $\ket{\varphi_1 \otimes \varphi_2}_{\mathcal{H}_1 {\otimes} \mathcal{H}_2}$ describing the state of a composite system is composed of the tensor product of the ket vectors $\ket{\varphi_i}_{\mathcal{H}_i}$ ($i=1,2$), each ket describing the state of a single particle.  
%
%
%

The obtained relations can be applied to a $N$-particle system ($N<\infty$).
The RHS comprises the $N$-multiple tensor product of RHS,  
represented by,
\begin{eqnarray}
    \widehat{\otimes}_{j=1}^N {\Phi_j} 
    \subset 
    \overline{\otimes}_{j=1}^N \mathcal{H}_j
    \subset 
    (\widehat{\otimes}_{j=1}^N \Phi_j)^{\prime},
    ~(\widehat{\otimes}_{j=1}^N \Phi_j)^{\times},
    \label{eqn:O2-2-1}
\end{eqnarray}
where
$\widehat{\otimes}_{j=1}^N {\Phi}_j=(\widehat{\otimes}_{j=1}^N {\Phi}_j, \widehat{\tau_{p}})$ is the tensor product obtained by completion of the algebraic tensor product $(\otimes_{j=1}^N {\Phi}_j, {\tau_{p}})$ of the nuclear spaces $(\Phi_j,\tau_{\Phi_j})~(j=1,\cdots, N)$.
$\overline{\otimes}_{j=1}^N \mathcal{H}_j=(\overline{\otimes}_{j=1}^N \mathcal{H}_j, \langle \cdot,\, \cdot\rangle_{\overline{\otimes}_{j=1}^N \mathcal{H}_j})$ is the tensor product space of Hilbert spaces whose inner product represents $\langle \cdot,\, \cdot\rangle_{\overline{\otimes}_{j=1}^N \mathcal{H}_j}$. 
The spaces $
(\widehat{\otimes}_{j=1}^N \Phi_j)^{\prime}$
and
$(\widehat{\otimes}_{j=1}^N \Phi_j)^{\times}$ 
are the dual and anti-dual spaces of $\widehat{\otimes}_{j=1}^N {\Phi}_j$, respectively.
Note that $\widehat{\otimes}_{j=1}^N {\Phi}_j$ becomes a nuclear space.
Using (\ref{eqn:O2-2-1}), the bra and ket vectors are defined as
\begin{eqnarray}
    \bra{\varphi}_{\overline{\otimes}_{j=1}^N \mathcal{H}_j} : \widehat{\otimes}_{j=1}^N {\Phi}_j \to \mathbb{C},~
    \bra{\varphi}_{\overline{\otimes}_{j=1}^N \mathcal{H}_j}(\phi)=
    \langle \varphi,\, \phi\rangle_{\overline{\otimes}_{j=1}^N \mathcal{H}_j}, \\
    \label{eqn:o2-2-2a}
    \ket{\varphi}_{\overline{\otimes}_{j=1}^N \mathcal{H}_j} : \widehat{\otimes}_{j=1}^N {\Phi}_j \to \mathbb{C},~
    \ket{\varphi}_{\overline{\otimes}_{j=1}^N \mathcal{H}_j}(\phi)=\langle \phi,\, \varphi\rangle_{\overline{\otimes}_{j=1}^N \mathcal{H}_j},    
    \label{eqn:o2-2-2b}
\end{eqnarray}
for $\varphi\in \widehat{\otimes}_{j=1}^N {\Phi_j}$. 
In this case, the relation (\ref{eqn:ket_relation}) becomes
\begin{eqnarray}
    \ket{\varphi_1 \otimes\dots \otimes \varphi_N}_{\overline{\otimes}_{j=1}^N \mathcal{H}_j}
    =
    \ket{\varphi_1}_{\mathcal{H}_1} \otimes \dots \otimes \ket{\varphi_N}_{\mathcal{H}_N}
\label{eqn:ket_relation_n}
\end{eqnarray}
in $
(\widehat{\otimes}_{j=1}^N \Phi_j)^{\times}$, where $\varphi_j\in \Phi_j,~j=1,\cdots, N$. 
Similarly to (\ref{eqn:ket_relation_n}), 
the following relation of the bra vectors in $
(\widehat{\otimes}_{j=1}^N \Phi_j)^{\prime}$ is derived:
\begin{eqnarray}
    \bra{\varphi_1 \otimes\dots \otimes \varphi_N}_{\overline{\otimes}_{j=1}^N \mathcal{H}_j}
    =
    \bra{\varphi_1}_{\mathcal{H}_1} \otimes \dots \otimes \bra{\varphi_N}_{\mathcal{H}_N}.
\label{eqn:bra_relation_n}
\end{eqnarray}
%


\subsection{Permutation operator on the dual spaces}
\label{sec:2.2}

The symmetry of identical particles in the Hilbert space theory can be introduced by using the permutation operator~\cite{Simon1980,Hall2013}.
Now we focus on the case where $\mathcal{H}_1=\mathcal{H}_2=\dots=\mathcal{H}_N\equiv \mathcal{H}$ 
and
$\Phi_1=\Phi_2=\dots=\Phi_N \equiv \Phi$.
Let $\mathfrak{S}_N$  be the symmetry group of degree $N$.
We fix $\sigma \in \mathfrak{S}_N$ and define the permutation, $U_{\sigma} : \otimes^{N} \mathcal{H} \to \overline{\otimes}^N \mathcal{H}$, on the algebraic tensor product $\otimes^{N} \mathcal{H}$
where 
\begin{eqnarray}
    U_\sigma (\phi)=\sum_{j=1}^m \phi_{\sigma(1)j}\otimes\dots \otimes \phi_{\sigma(N)j}
    ~~
    \mbox{for}
    ~~
    \phi=\sum_{j=1}^m \phi_{1j}\otimes\dots \otimes \phi_{Nj}\in \otimes^{N} \mathcal{H}.
    \label{eqn:O3-0}
\end{eqnarray}
%
%
The permutation has the unique extension to the completion $(\overline{\otimes}^N \mathcal{H}, \langle \cdot,\, \cdot\rangle_{\overline{\otimes}^N \mathcal{H}})$ of the inner product space $(\otimes^N \mathcal{H}, \langle \cdot,\, \cdot\rangle_{{\otimes}^N \mathcal{H}})$.
We denote this extension by $U_{\sigma}$.
Corresponding to this case, the following triplet of the $N$-tensor product space of RHS is adapted, similar to that of (\ref{eqn:O2-2-1}),
\begin{eqnarray}
    \widehat{\otimes}^N {\Phi} 
    \subset \overline{\otimes}^N \mathcal{H}
    \subset 
    (\widehat{\otimes}^N \Phi)^{\prime},
    ~(\widehat{\otimes}^N \Phi)^{\times}.
    \label{eqn:O3-1}
\end{eqnarray}
%
%
%
The permutation suitable for the RHS (\ref{eqn:O3-1}) can be established as follows. 
Let the permutation $U_{\sigma}$ on $\otimes^{N}\mathcal{H}$ be restricted to the algebraic tensor product $\otimes^N \Phi$.
%
%
Consequently, the restriction $ U_\sigma|_{ \otimes^N \Phi}$ becomes an isomorphism of $\otimes^N \Phi$ onto itself, with respect to the nuclear topology $\tau_p$~\cite{Maurin1968}.
Therefore, to the nuclear space $(\widehat{\otimes}^N {\Phi}, \widehat{\tau_{p}})$, there exists the unique extension $U_{\sigma}^{\widehat{\otimes}^N {\Phi}}$ of $ U_\sigma|_{\otimes^N \Phi}$.
The uniqueness of $U_{\sigma}^{\widehat{\otimes}^N {\Phi}}$ shows $U_{\sigma}^{\widehat{\otimes}^N {\Phi}}=U_\sigma|_ {\widehat{\otimes}^N {\Phi}}$, and hence we obtain the permutation on the nuclear space $\widehat{\otimes}^N {\Phi}$ in the form of
\begin{eqnarray}
    U_{\sigma}^{\widehat{\otimes}^N {\Phi}} : (\widehat{\otimes}^N {\Phi}, \widehat{\tau_{p}}) \to (\widehat{\otimes}^N {\Phi}, \widehat{\tau_{p}}),~\phi \mapsto U_\sigma(\phi).    
    \label{eqn:o3-2}
\end{eqnarray}
Note that $U_{\sigma}^{\widehat{\otimes}^N {\Phi}}$ is an isomorphism  
and the relation 
\begin{eqnarray}
    i \circ U_{\sigma}^{\widehat{\otimes}^N {\Phi}} = U_\sigma \circ i    
    \label{eqn:o3-3}
\end{eqnarray}
is satisfied where $i$ is the canonical embedding that characterizes the RHS (\ref{eqn:O3-1}).

The symmetric structure for the tensor product of the Hilbert space, $\overline{\otimes}^N \mathcal{H}$, is characterized by the following projection, referred to as the permutation operator~\cite{Simon1980},
\begin{eqnarray}
    P_c=\frac{1}{N!}\displaystyle\sum_{\sigma\in \mathfrak{S}_N} c(\sigma)U_\sigma.      
    \label{eqn:projection1}
\end{eqnarray}
Similarly, the projection for the nuclear space $\widehat{\otimes}^N {\Phi}$ is introduced by using 
$U_{\sigma}^{\widehat{\otimes}^N {\Phi}}$ as
\begin{eqnarray}
       P_c^{\widehat{\otimes}^N {\Phi}}
    =
    \frac{1}{N!}\displaystyle\sum_{\sigma\in \mathfrak{S}_N} c(\sigma)U_\sigma^{\widehat{\otimes}^N {\Phi}}.
    \label{eqn:projection2}
\end{eqnarray}
%

%

The bra and ket vectors constructed using the tensor product of RHS belong to the dual spaces, as shown in the previous subsection.
This fact necessitates the extension of the permutation operator (\ref{eqn:projection2}) 
to the dual spaces.
As the operator (\ref{eqn:projection2}) is continuous on $\widehat{\otimes}^N \Phi$ and maps onto $\widehat{\otimes}^N \Phi$,
the extension of (\ref{eqn:projection2}) can be easily constructed as follows.
We set a operator $\widetilde{P_c^{\widehat{\otimes}^N {\Phi}}}$ on $
    (\widehat{\otimes}^N \Phi)^{\prime}\cup
    (\widehat{\otimes}^N \Phi)^{\times}$, 
    where  
\begin{eqnarray}
       \widetilde{P_c^{\widehat{\otimes}^N {\Phi}}}(f)(\phi)
    =f(P_c^{\widehat{\otimes}^N {\Phi}}(\phi)),
    \label{eqn:o3-2-1}
\end{eqnarray}
for $f\in  (\widehat{\otimes}^N \Phi)^{\prime}\cup
    (\widehat{\otimes}^N \Phi)^{\times}$, $\phi \in \widehat{\otimes}^N \Phi$.
This operator (\ref{eqn:o3-2-1}) endows the symmetric structure for the bra and ket vectors satisfying (\ref{eqn:ket_relation_n}) and (\ref{eqn:bra_relation_n}).
To show this fact, we fixed $N=2$ in short.
In the nuclear space $(\widehat{\otimes}^2 \Phi,\widehat{\tau_p})$, 
each $\phi \in \widehat{\otimes}^2 \Phi$ can be represented as the form of the sum of an absolutely convergent series, $\phi =\sum_{i=1}^{\infty}\lambda_i \phi_i^1\otimes \phi_i^2$,
where $\sum_{i}|\lambda_i|\leq 1$ and $\{\phi_i^1\}$ and $\{\phi_i^2\}$ are null sequences in $\Phi$~\cite{Schaefer1966}.
As $P_c^{\widehat{\otimes}^2 {\Phi}}$ is continuous linear on $(\widehat{\otimes}^2\Phi,\widehat{\tau_p})$,
we have
\begin{eqnarray}
    P_c^{\widehat{\otimes}^2 {\Phi}}(\phi) & = &
    \sum_{i=1}^{\infty}\lambda_i P_c^{\widehat{\otimes}^2 {\Phi}}(\phi_i^1\otimes \phi_i^2)\nonumber\\
     & = &
     \left\{
    \begin{array}{ll}
    \displaystyle\sum_{i=1}^{\infty} \frac{\lambda_i}{2}(\phi_i^1\otimes \phi_i^2
    +\phi_i^2\otimes \phi_i^1) 
    & (c=c_1) \\
    \displaystyle\sum_{i=1}^{\infty} \frac{\lambda_i}{2}(\phi_i^1\otimes \phi_i^2
    -\phi_i^2\otimes \phi_i^1) 
    & (c=sgn)
    \end{array}
    \right.
    \label{eqn:o3-2-3}
\end{eqnarray}
Here, we focused on the symmetry case, $c=c_1$.
(in the same manner, the anti-symmetric case is also obtained.)
For $\varphi=\varphi_1 \otimes \varphi_2 \in \otimes^2 \Phi \subset \widehat{\otimes}^2 \Phi$,
from (\ref{eqn:ket_relation}),
$\ket{\varphi}_{\overline{\otimes}^2 \mathcal{H}}
    =
    \ket{\varphi_1 \otimes \varphi_2}_{\overline{\otimes}^2 \mathcal{H}}
    =
    \ket{\varphi_1}_{\mathcal{H}} \otimes \ket{\varphi_2}_{\mathcal{H}}.
$
By using (\ref{eqn:o3-2-3}) and by considering the continuity and anti-linearity of a ket acting on $(\widehat{\otimes}^2 \Phi,\widehat{\tau_p})$, through the calculation,
\begin{eqnarray}
    \widetilde{P_c^{\widehat{\otimes}^2 {\Phi}}}(\ket{\varphi_1 \otimes \varphi_2}_{\overline{\otimes}^2 \mathcal{H}})(\phi)
    & = &
    \ket{\varphi_1 \otimes \varphi_2}_{\overline{\otimes}^2 \mathcal{H}}(P_{c_1}^{\widehat{\otimes}^2 {\Phi}}(\phi))
    \nonumber \\
    & = &
        \ket{\varphi_1 \otimes \varphi_2}_{\overline{\otimes}^2 \mathcal{H}}
    \Bigr\{ \sum_{i=1}^{\infty} \frac{\lambda_i}{2}(\phi_i^1\otimes \phi_i^2
    +\phi_i^2\otimes \phi_i^1)\Bigr\}
     \nonumber\\
    & = &
    \sum_{i=1}^{\infty} \frac{\lambda^*_i}{2}
    \Bigr\{
    \ket{\varphi_1 \otimes \varphi_2}_{\overline{\otimes}^2 \mathcal{H}}(\phi_i^1\otimes \phi_i^2)
    +\ket{\varphi_1 \otimes \varphi_2}_{\overline{\otimes}^2 \mathcal{H}}(\phi_i^2\otimes \phi_i^1)
    \Bigr\}
    \nonumber\\
     & = &
    \sum_{i=1}^{\infty} \frac{\lambda^*_i}{2}
    \Bigr\{
    \langle \phi_i^1\otimes \phi_i^2
    ,\, 
    \varphi_1 \otimes \varphi_2\rangle_{\overline{\otimes}^2 \mathcal{H}}
    +
    \langle \phi_i^2\otimes \phi_i^1
    ,\, 
    \varphi_1 \otimes \varphi_2\rangle_{\overline{\otimes}^2 \mathcal{H}}
    \Bigr\}
    \nonumber\\
    & = &
    \sum_{i=1}^{\infty} \frac{\lambda^*_i}{2}
    \Bigr\{
    \langle \phi_i^1
    ,\, 
    \varphi_1\rangle_{\mathcal{H}}
    \langle \phi_i^2
    ,\, 
    \varphi_2\rangle_{\mathcal{H}}
    +
    \langle \phi_i^2
    ,\, 
    \varphi_1\rangle_{\mathcal{H}}
    \langle \phi_i^1
    ,\, 
    \varphi_2\rangle_{\mathcal{H}}
    \Bigr\}
    \nonumber\\
    & = &
    \sum_{i=1}^{\infty} \frac{\lambda^*_i}{2}
    \Bigr\{
    \langle \phi_i^1\otimes \phi_i^2
    ,\, 
    \varphi_1 \otimes \varphi_2\rangle_{\overline{\otimes}^2 \mathcal{H}}
    +
    \langle \phi_i^1\otimes \phi_i^2
    ,\, 
    \varphi_2 \otimes \varphi_1\rangle_{\overline{\otimes}^2 \mathcal{H}}
    \Bigr\}
    \nonumber\\
    & = &
    \frac{1}{2}
    (
    \ket{\varphi_1}_\mathcal{H}
    \otimes
    \ket{\varphi_2}_\mathcal{H}
    +
    \ket{\varphi_2}_\mathcal{H}
    \otimes
    \ket{\varphi_1}_\mathcal{H}
    )
    (\sum_{i=1}^{\infty}\lambda_i \phi_i^1\otimes \phi_i^2)
    \nonumber\\
    & = &
    \frac{1}{2}
    (
    \ket{\varphi_1}_\mathcal{H}
    \otimes
    \ket{\varphi_2}_\mathcal{H}
    +
    \ket{\varphi_2}_\mathcal{H}
    \otimes
    \ket{\varphi_1}_\mathcal{H}
    )(\phi), 
    \label{eqn:o3-calculation1}
\end{eqnarray}
we have the relation $\widetilde{P_c^{\widehat{\otimes}^2 {\Phi}}}(\ket{\varphi_1 \otimes \varphi_2}_{\overline{\otimes}^2 \mathcal{H}})(\phi)=\frac{1}{2}
    (
    \ket{\varphi_1}_\mathcal{H}
    \otimes
    \ket{\varphi_2}_\mathcal{H}
    +
    \ket{\varphi_2}_\mathcal{H}
    \otimes
    \ket{\varphi_1}_\mathcal{H}
    )(\phi)$ for any $\phi \in \widehat{\otimes}^2 \Phi$. 
    Therefore, the symmetric relation for the ket vectors in $(\widehat{\otimes}^2 \Phi)^{\times}$ is obtained :
\begin{eqnarray}
    \widetilde{P_c^{\widehat{\otimes}^2 {\Phi}}}(\ket{\varphi_1}_\mathcal{H}
    \otimes
    \ket{\varphi_2}_\mathcal{H})
     = 
    \frac{1}{2}(\ket{\varphi_1}_\mathcal{H}
    \otimes
    \ket{\varphi_2}_\mathcal{H}
    +
    \ket{\varphi_2}_\mathcal{H}
    \otimes
    \ket{\varphi_1}_\mathcal{H}).
    \label{eqn:o3-2-4}
\end{eqnarray}
This relation can be generalized to the $N$-tensor product case; 
for $\ket{\varphi_1,\otimes\dots \otimes \varphi_N}_{\overline{\otimes}^N\mathcal{H}}=\ket{\varphi_1}_{\mathcal{H}}\otimes \dots\otimes \ket{\varphi_N}_{\mathcal{H}}$ in $(\widehat{\otimes}^N \Phi)^{\times}$
where $\varphi_1\otimes\dots \otimes \varphi_N\in \widehat{\otimes}^N \Phi$, we have
\begin{eqnarray}
    \widetilde{P_c^{\widehat{\otimes}^2 {\Phi}}}
    (\ket{\varphi_1}_\mathcal{H}
    \otimes\dots\otimes
    \ket{\varphi_N}_\mathcal{H})
    =
     \left\{
    \begin{array}{ll}
    \displaystyle\frac{1}{N!}\sum_{\sigma \in \mathfrak{S}_n}\ket{\varphi_{\sigma(1)}}_{\mathcal{H}}
    \otimes \dots \otimes
    \ket{\varphi_{\sigma(N)}}_{\mathcal{H}}
    & (c=c_1) \\
    \displaystyle\frac{1}{N!}\sum_{\sigma \in \mathfrak{S}_n}sgn(\sigma)\ket{\varphi_{\sigma(1)}}_{\mathcal{H}}
    \otimes \dots \otimes
    \ket{\varphi_{\sigma(N)}}_{\mathcal{H}}
    & (c=sgn).
    \end{array}
    \right.
    \label{eqn:o3-ket-symmetry}
\end{eqnarray}
Here, (\ref{eqn:o3-ket-symmetry}) presents the symmetry and anti-symmetry for only the ket vectors of in the space $(\widehat{\otimes}^2 \Phi)^{\times}$.
Related to (\ref{eqn:o3-ket-symmetry}), 
we set the spaces
\begin{eqnarray}
        (\widehat{\otimes}^N \Phi)_s^{\times}=\widetilde{P_{c_1}^{\widehat{\otimes}^2 {\Phi}}}
    ((\widehat{\otimes}^N \Phi)^{\times}),
    \label{eqn:ket-symm-setA}
    \\
        (\widehat{\otimes}^N \Phi)_a^{\times}=\widetilde{P_{sgn}^{\widehat{\otimes}^2 {\Phi}}}
    ((\widehat{\otimes}^N \Phi)^{\times}),
    \label{eqn:ket-symm-setB}
\end{eqnarray}
and refer to them as the symmetric and anti-symmetric ket spaces, respectively. 

In terms of $(\widehat{\otimes}^2 \Phi)^{\prime}$,
the symmetric structure for the bra vector is expressed using the permutation operator (\ref{eqn:o3-2-1}), as follows,
\begin{eqnarray}
    \widetilde{P_c^{\widehat{\otimes}^2 {\Phi}}}
    (\bra{\varphi_1}_\mathcal{H}
    \otimes\dots\otimes
    \bra{\varphi_N}_\mathcal{H})
    =
     \left\{
    \begin{array}{ll}
    \displaystyle\frac{1}{N!}\sum_{\sigma \in \mathfrak{S}_n}\bra{\varphi_{\sigma(1)}}_{\mathcal{H}}
    \otimes \dots \otimes
    \bra{\varphi_{\sigma(N)}}_{\mathcal{H}}
    & (c=c_1) \\
    \displaystyle\frac{1}{N!}\sum_{\sigma \in \mathfrak{S}_n}sgn(\sigma)\bra{\varphi_{\sigma(1)}}_{\mathcal{H}}
    \otimes \dots \otimes
    \bra{\varphi_{\sigma(N)}}_{\mathcal{H}}
    & (c=sgn).
    \end{array}
    \right.
    \label{eqn:o3-bra-symmetry}
\end{eqnarray}
Further, the symmetric and the anti-symmetric bra spaces are expressed as the following sets, respectively : 
\begin{eqnarray}
        (\widehat{\otimes}^N \Phi)_s^{\prime}=\widetilde{P_{c_1}^{\widehat{\otimes}^2 {\Phi}}}
    ((\widehat{\otimes}^N \Phi)^{\prime}),
    \label{eqn:bra-symm-setA}
    \\
        (\widehat{\otimes}^N \Phi)_a^{\prime}=\widetilde{P_{sgn}^{\widehat{\otimes}^2 {\Phi}}}
    ((\widehat{\otimes}^N \Phi)^{\prime}).
    \label{eqn:bra-symm-setB}
\end{eqnarray}
Thus, in the RHS formulation that characterizes the identical particle system, 
the symmetric structure can be individually assigned to the bra and ket vectors.

When we combine the dual spaces $(\widehat{\otimes}^N \Phi)^{\prime}$ and $(\widehat{\otimes}^N \Phi)^{\times}$ as $(\widehat{\otimes}^N \Phi)^{\prime}\cup(\widehat{\otimes}^N \Phi)^{\times}$,
the symmetric and anti-symmetric spaces of $(\widehat{\otimes}^N \Phi)^{\prime}\cup(\widehat{\otimes}^N \Phi)^{\times}$ become  
\begin{eqnarray}
    \big{[}(\widehat{\otimes}^N \Phi)^{\prime}\cup(\widehat{\otimes}^N \Phi)^{\times}\big{]}_s
   & = &
    \widetilde{P_{c_1}^{\widehat{\otimes}^2 {\Phi}}}((\widehat{\otimes}^N \Phi)^{\prime}\cup(\widehat{\otimes}^N \Phi)^{\times})
    =
    \widetilde{P_{c_1}^{\widehat{\otimes}^2 {\Phi}}}
    ((\widehat{\otimes}^N \Phi)^{\prime})\cup 
    \widetilde{P_{c_1}^{\widehat{\otimes}^2 {\Phi}}}
    ((\widehat{\times}^N \Phi)^{\times})
    \nonumber\\
    & = & (\widehat{\otimes}^N \Phi)_s^{\prime}\cup (\widehat{\otimes}^N \Phi)_s^{\times}.
    \label{eqn:braket-symm-setA}
\end{eqnarray}
and 
\begin{eqnarray}
    \big{[}(\widehat{\otimes}^N \Phi)^{\prime}\cup(\widehat{\otimes}^N \Phi)^{\times}\big{]}_a
    =
    (\widehat{\otimes}^N \Phi)_a^{\prime}\cup (\widehat{\otimes}^N \Phi)_a^{\times},
    \label{eqn:braket-symm-setB}
\end{eqnarray}
respectively.

\section{Observable}
\label{sec:3}

\subsection{Spectral expansion in the tensor product of RHS}
\label{sec:3.2}

Now we consider on an self-adjoint operator with respect to the tensor product of RHS (\ref{eqn:O2-4}) and its spectral decomposition based on RHS approach.
We set $N=2$ for simplicity.
Let $A_i : D(A_i)\rightarrow \mathcal{H}_i$ be self-adjoint in $\mathcal{H}_i$ where $D(A_i)$ indicates the domain of $A_i$ $(i=1,2)$.
Each $A_i$ is assumed to be continuous on $\Phi_i$, satisfying $A_i(\Phi_i)\subset \Phi_i$.
Now, we focus on a self-adjoint operator defined in the tensor product $\mathcal{H}_1 \overline{\otimes}\mathcal{H}_2$,
\begin{eqnarray}
    A=\overline{A_1\otimes I_2 +I_1\otimes A_2} : D({A}) \to \mathcal{H}_1 \overline{\otimes}\mathcal{H}_2, 
    \label{composed_operator}
\end{eqnarray}
which is given by the self-adjoint extension of the operator $A_1\otimes I_2 +I_1\otimes A_2$ in $\mathcal{H}_1 \overline{\otimes}\mathcal{H}_2$ where $I_i$ is the identity map for $\mathcal{H}_i$ $(i=1,2)$.
Notably, this form of $A$ is generally utilized as the Hamiltonian of a composite system~\cite{Simon1980,Messiah}.
This operator has the spectrum $Sp(A)=Cl(Sp(A_1)+Sp(A_2))$ lying on the real line
($ClX$ is the closure of a set $X$ in the real line).
Also, it is known that $A$ is continuous on the nuclear space $\Phi_1 \widehat{\otimes} \Phi_2$ such that the relation ${A}(\Phi_1 \widehat{\otimes} \Phi_2)\subset \Phi_1 \widehat{\otimes} \Phi_2$ holds~\cite{Maurin1968}.
Therefore, by the nuclear spectral theorem for the self-adjoint operator $A$ of the form (\ref{composed_operator}), the following relations are obtained~\cite{Maurin1968} :
for any $\varphi, \psi \in \Phi_1 \widehat{\otimes} \Phi_2$,

\begin{eqnarray}
    \langle \varphi,\, \psi\rangle_{\mathcal{H}_1\overline{\otimes} \mathcal{H}_2}
    & = & 
     \displaystyle
     \int_{\lambda\in Sp({A})}
     \braket{\hat{\varphi}}{\hat{\psi}}_\lambda
     d\mu_\lambda,
    \label{eqn:o3-2-1a}\\
    \langle \varphi,\, A\psi\rangle_{\mathcal{H}_1\overline{\otimes} \mathcal{H}_2}
    & = & 
     \displaystyle
     \int_{\lambda\in Sp({A})}
     \lambda\braket{\hat{\varphi}}{\hat{\psi}}_\lambda
     d\mu_\lambda,
    \label{eqn:o3-2-1b}
\end{eqnarray}
with
\begin{eqnarray}
    \displaystyle\braket{\hat{\varphi}}{\hat{\psi}}_\lambda
    =
    \displaystyle\int_{\lambda=\lambda_1+\lambda_2}
     \sum_{k=1}^{dim \hat{\mathcal{H}_1}(\lambda_1)} 
     \sum_{l=1}^{dim \hat{\mathcal{H}_2}(\lambda_2)}
     (e_{\lambda_1,k}^1\otimes
     e_{\lambda_2,k}^2)^*(\varphi)
     (e_{\lambda_1,k}^1\otimes
     e_{\lambda_2,l}^2)(\psi)
     d\sigma^\lambda_{\lambda_1,\lambda_2},
    \label{eqn:o3-2-1c}
\end{eqnarray}
where $\mu_\lambda$ is the Borel measure, 
$\sigma^\lambda_{\lambda_1,\lambda_2}$ is also a Borel measure on $\mathbb{R}^2$ whose support is contained in the set $\{(\lambda_1,\lambda_2)\in \mathbb{R}^2 ; \lambda=\lambda_1+\lambda_2, \lambda_i\in Sp(A_i) (i=1,2)\}$.
$e_{\lambda_1,k}^1(k=1,2,\cdots, dim \hat{\mathcal{H}_1}(\lambda_1))$ and $e_{\lambda_2,l}^2~(l=1,2,\cdots, dim \hat{\mathcal{H}_2}(\lambda_2))$ are the generalized eigenvectors of $A_1$ and $A_2$ corresponding to $\lambda_1$ and $\lambda_2$ respectively. 
When $dim \hat{\mathcal{H}_1}(\lambda_1))=dim \hat{\mathcal{H}_2}(\lambda_2)=1$,
the relations ($\ref{eqn:o3-2-1a}$) and ($\ref{eqn:o3-2-1b}$) are expressed as
\begin{eqnarray}
    \langle \varphi,\, \psi\rangle_{\mathcal{H}_1\overline{\otimes} \mathcal{H}_2}
    & = & 
     \displaystyle
     \int_{\lambda\in Sp({A})}
    \Big{\{}
    \displaystyle\int_{\lambda=\lambda_1+\lambda_2}
     (e^1_{\lambda_1}\otimes
     e^2_{\lambda_2})^*(\varphi)
     (e^1_{\lambda_1}\otimes
     e^2_{\lambda_2})(\psi)
     d\sigma^\lambda_{\lambda_1,\lambda_2}
     \Big{\}}
     d\mu_\lambda,
    \label{eqn:o3-2-2a}
    \\
    \langle \varphi,\, A\psi\rangle_{\mathcal{H}_1\overline{\otimes} \mathcal{H}_2}
    & = & 
     \displaystyle
     \int_{\lambda\in Sp({A})}
    \lambda
    \Big{\{}
    \displaystyle\int_{\lambda=\lambda_1+\lambda_2}
     (e^1_{\lambda_1}\otimes
     e^2_{\lambda_2})^*(\varphi)
     (e^1_{\lambda_1}\otimes
     e^2_{\lambda_2})(\psi)
     d\sigma^\lambda_{\lambda_1,\lambda_2}
     \Big{\}}
     d\mu_\lambda.
    \label{eqn:o3-2-2b}
\end{eqnarray}
When the following notations are introduced,
\begin{eqnarray}
    e^i_{\lambda_i}\to \bra{\lambda_i}_{\mathcal{H}_i}, ~~
    (e^i_{\lambda_i})^*\to \ket{\lambda_i}_{\mathcal{H}_i}, ~~(i=1,2)
    \label{notations1}
\end{eqnarray}
and 
\begin{eqnarray}
     e^1_{\lambda_1}\otimes
     e^2_{\lambda_2}(\varphi)
     \to
    \bra{\lambda_1}_{\mathcal{H}_1}\otimes
     \bra{\lambda_2}_{\mathcal{H}_2} \ket{\varphi}_{\mathcal{H}_1\overline{\otimes}\mathcal{H}_2},~
     \nonumber\\
    (e^1_{\lambda_1}\otimes
     e^2_{\lambda_2})^*(\varphi)
     \to
    \bra{\varphi}_{\mathcal{H}_1\overline{\otimes}\mathcal{H}_2}
    \ket{\lambda_1}_{\mathcal{H}_1}\otimes
     \ket{\lambda_2}_{\mathcal{H}_2},
    \label{notations2}
\end{eqnarray}
%
%
(\ref{eqn:o3-2-2a}) and (\ref{eqn:o3-2-2b}) are represented as
\begin{align}
    \begin{split} 
    \langle \varphi,\, \psi\rangle_{\mathcal{H}_1\overline{\otimes} \mathcal{H}_2}
    & = 
    \displaystyle
     \int_{\lambda\in Sp({A})}
     \Big{\{}
    \displaystyle\int_{\lambda=\lambda_1+\lambda_2}
     \bra{\varphi}_{\mathcal{H}_1\overline{\otimes}\mathcal{H}_2}\ket{\lambda_1}_{\mathcal{H}_1}\otimes
     \ket{\lambda_2}_{\mathcal{H}_2}\\
     &\bra{\lambda_1}_{\mathcal{H}_1}\otimes
     \bra{\lambda_2}_{\mathcal{H}_2} \ket{\psi}_{\mathcal{H}_1\overline{\otimes}\mathcal{H}_2}
    d\sigma^\lambda_{\lambda_1,\lambda_2}
     \Big{\}}
     d\mu_\lambda,
     \end{split}
    \label{eqn:o3-2-3a}
\end{align}
and 
\begin{align}
    \begin{split} 
    \langle \varphi,\, A\psi\rangle_{\mathcal{H}_1\overline{\otimes} \mathcal{H}_2}
    & = 
    \displaystyle
     \int_{\lambda\in Sp({A})}
     \lambda
     \Big{\{}
    \displaystyle\int_{\lambda=\lambda_1+\lambda_2}
     \bra{\varphi}_{\mathcal{H}_1\overline{\otimes}\mathcal{H}_2}\ket{\lambda_1}_{\mathcal{H}_1}\otimes
     \ket{\lambda_2}_{\mathcal{H}_2}\\
     &\bra{\lambda_1}_{\mathcal{H}_1}\otimes
     \bra{\lambda_2}_{\mathcal{H}_2} \ket{\psi}_{\mathcal{H}_1\overline{\otimes}\mathcal{H}_2}
    d\sigma^\lambda_{\lambda_1,\lambda_2}
     \Big{\}}
     d\mu_\lambda,
     \end{split}
    \label{eqn:o3-2-3b}
\end{align}
for any $\varphi, \psi \in \Phi_1 \widehat{\otimes} \Phi_2$.
Note that $\bra{\lambda_1}_{\mathcal{H}_1}\otimes
     \bra{\lambda_2}_{\mathcal{H}_2}$ and $\ket{\lambda_1}_{\mathcal{H}_1}\otimes
     \ket{\lambda_2}_{\mathcal{H}_2}$ 
     belong to $(\Phi_1 \widehat{\otimes} \Phi_2)^\prime$ and 
     $(\Phi_1 \widehat{\otimes} \Phi_2)^{\times}$, 
     respectively.

The representations (\ref{eqn:o3-2-3a}) and (\ref{eqn:o3-2-3b}) indicate the expansions that are performed based on the tensor products of generalized eigenvectors
$\{\bra{\lambda_1}_{\mathcal{H}_1}\otimes
     \bra{\lambda_2}_{\mathcal{H}_2}\}$ and $\{\ket{\lambda_1}_{\mathcal{H}_1}\otimes
     \ket{\lambda_2}_{\mathcal{H}_2}\}$, as follows,
\begin{align}
    \begin{split} 
     \ket{\varphi}_{\mathcal{H}_1\overline{\otimes} \mathcal{H}_2}
    & = 
    \displaystyle\int_{\lambda\in Sp(A)}
    \displaystyle\int_{\lambda=\lambda_1+\lambda_2}
     \bra{\lambda_1}_{\mathcal{H}_1}\otimes
     \bra{\lambda_2}_{\mathcal{H}_2} \ket{\varphi}_{\mathcal{H}_1\overline{\otimes}\mathcal{H}_2}    \ket{\lambda_1}_{\mathcal{H}_1}\otimes
     \ket{\lambda_2}_{\mathcal{H}_2}
    d\sigma^\lambda_{\lambda_1,\lambda_2} d\mu_\lambda,
    \end{split}
    \label{spectralexpansion_ket_a}
    \\
    \begin{split}
    \ket{A\varphi}_{\mathcal{H}_1\overline{\otimes} \mathcal{H}_2}
    & = 
    \displaystyle\int_{\lambda\in Sp(A)}\lambda
    \displaystyle\int_{\lambda=\lambda_1+\lambda_2}
     \bra{\lambda_1}_{\mathcal{H}_1}\otimes
     \bra{\lambda_2}_{\mathcal{H}_2} \ket{\varphi}_{\mathcal{H}_1\overline{\otimes}\mathcal{H}_2}
     \ket{\lambda_1}_{\mathcal{H}_1}\otimes
     \ket{\lambda_2}_{\mathcal{H}_2}
    d\sigma^\lambda_{\lambda_1,\lambda_2} d\mu_\lambda,
    \end{split}
    \label{spectralexpansion_ket_b}
\end{align}
\begin{align}
    \begin{split} 
     \bra{\varphi}_{\mathcal{H}_1\overline{\otimes} \mathcal{H}_2}
    & = 
    \displaystyle\int_{\lambda\in Sp(A)}
    \displaystyle\int_{\lambda=\lambda_1+\lambda_2}
     \bra{\varphi}_{\mathcal{H}_1\overline{\otimes}\mathcal{H}_2}    \ket{\lambda_1}_{\mathcal{H}_1}\otimes
     \ket{\lambda_2}_{\mathcal{H}_2}
     \bra{\lambda_1}_{\mathcal{H}_1}\otimes
     \bra{\lambda_2}_{\mathcal{H}_2}
    d\sigma^\lambda_{\lambda_1,\lambda_2} d\mu_\lambda,
    \end{split}
    \label{spectralexpansion_bra_a}
    \\
    \begin{split}
    \bra{A\varphi}_{\mathcal{H}_1\overline{\otimes} \mathcal{H}_2}
    & = 
    \displaystyle\int_{\lambda\in Sp(A)}\lambda
    \displaystyle\int_{\lambda=\lambda_1+\lambda_2}
      \bra{\varphi}_{\mathcal{H}_1\overline{\otimes}\mathcal{H}_2}
     \ket{\lambda_1}_{\mathcal{H}_1}\otimes
     \ket{\lambda_2}_{\mathcal{H}_2}
     \bra{\lambda_1}_{\mathcal{H}_1}\otimes
     \bra{\lambda_2}_{\mathcal{H}_2}
    d\sigma^\lambda_{\lambda_1,\lambda_2} d\mu_\lambda.
    \end{split}
    \label{spectralexpansion_bra_b}
\end{align}
Hereafter, we adopt the sign $\int_{Sp(A)}d\nu$ in stead of $\int_{\lambda\in Sp(A)d}d\sigma^\lambda_{\lambda_1,\lambda_2}\int_{\lambda=\lambda_1+\lambda_2}d\mu_\lambda$. Consequently, the relations obtained till now can be represented simply by using the abbreviation,
\begin{eqnarray}
    \int_{\lambda\in Sp(A)}\int_{\lambda=\lambda_1+\lambda_2}\to \int_{Sp(A)}~~ 
    \text{and}~~
    d\sigma^\lambda_{\lambda_1,\lambda_2}d\mu_\lambda \to d\nu.
    \label{abbreviation}
\end{eqnarray}
Then, the spectral expansions of $\ket{\varphi}_{\mathcal{H}_1\overline{\otimes} \mathcal{H}_2}$ and $\bra{\varphi}_{\mathcal{H}_1\overline{\otimes} \mathcal{H}_2}$ of (\ref{spectralexpansion_ket_a})--(\ref{spectralexpansion_bra_b}) convert into
\begin{eqnarray}
     \ket{\varphi}_{\mathcal{H}_1\overline{\otimes} \mathcal{H}_2}
    & = 
    \displaystyle\int_{Sp(A)}
     \bra{\lambda_1}_{\mathcal{H}_1}\otimes
     \bra{\lambda_2}_{\mathcal{H}_2} \ket{\varphi}_{\mathcal{H}_1\overline{\otimes}\mathcal{H}_2}    \ket{\lambda_1}_{\mathcal{H}_1}\otimes
     \ket{\lambda_2}_{\mathcal{H}_2}
    d\nu,\\
    \label{spectralexpansion_ket_aa}
         \ket{A\varphi}_{\mathcal{H}_1\overline{\otimes} \mathcal{H}_2}
    & = 
    \displaystyle\int_{Sp(A)}\lambda
     \bra{\lambda_1}_{\mathcal{H}_1}\otimes
     \bra{\lambda_2}_{\mathcal{H}_2} \ket{\varphi}_{\mathcal{H}_1\overline{\otimes}\mathcal{H}_2}    \ket{\lambda_1}_{\mathcal{H}_1}\otimes
     \ket{\lambda_2}_{\mathcal{H}_2}
    d\nu,
    \label{spectralexpansion_ket_ab}
    \end{eqnarray}
and
    \begin{eqnarray}
    \bra{\varphi}_{\mathcal{H}_1\overline{\otimes} \mathcal{H}_2}
    & = 
    \displaystyle\int_{Sp(A)}
     \bra{\varphi}_{\mathcal{H}_1\overline{\otimes}\mathcal{H}_2}    \ket{\lambda_1}_{\mathcal{H}_1}\otimes
     \ket{\lambda_2}_{\mathcal{H}_2}
     \bra{\lambda_1}_{\mathcal{H}_1}\otimes
     \bra{\lambda_2}_{\mathcal{H}_2} 
    d\nu,\\
    \label{spectralexpansion_bra_aa}
    \bra{A\varphi}_{\mathcal{H}_1\overline{\otimes} \mathcal{H}_2}
    & = 
    \displaystyle\int_{Sp(A)}\lambda
     \bra{\varphi}_{\mathcal{H}_1\overline{\otimes}\mathcal{H}_2}    \ket{\lambda_1}_{\mathcal{H}_1}\otimes
     \ket{\lambda_2}_{\mathcal{H}_2}
     \bra{\lambda_1}_{\mathcal{H}_1}\otimes
     \bra{\lambda_2}_{\mathcal{H}_2} 
    d\nu.
    \label{spectralexpansion_bra_ab}
\end{eqnarray}

When $\varphi=\varphi_1\otimes \varphi_2\in \Phi_1 \otimes \Phi_2$, the relation,
$\bra{\lambda_1}_\mathcal{H}\otimes\bra{\lambda_2}_\mathcal{H}(\varphi_1\otimes\varphi_2)
=\bra{\lambda_1}_\mathcal{H}(\varphi_1)\otimes\bra{\lambda_2}_\mathcal{H}(\varphi_2)=
\braket{\lambda_1}{\varphi_1}_{\mathcal{H}_1}\braket{\lambda_2}{\varphi_2}_{\mathcal{H}_2}$, can be utilized to obtain the spectral expansions of
$\ket{\varphi_1}_{\mathcal{H}_1}\otimes\ket{\varphi_2}_{\mathcal{H}_2}$ :
\begin{align}
\begin{split}
    \ket{\varphi_1}_{\mathcal{H}_1}\otimes\ket{\varphi_2}_{\mathcal{H}_2}
    & = \ket{\varphi_1\otimes\varphi_2}_{\mathcal{H}_1{\otimes} \mathcal{H}_2}
    =
    \ket{\varphi_1\otimes\varphi_2}_{\mathcal{H}_1\overline{\otimes} \mathcal{H}_2} \\
    & = 
    \displaystyle\int_{Sp(A)}
     \bra{\lambda_1}_{\mathcal{H}_1}\otimes
     \bra{\lambda_2}_{\mathcal{H}_2} 
    \ket{\varphi_1\otimes\varphi_2}_{\mathcal{H}_1\overline{\otimes} \mathcal{H}_2}     
     \ket{\lambda_1}_{\mathcal{H}_1}\otimes
     \ket{\lambda_2}_{\mathcal{H}_2}
    d\nu\\
    & =
    \displaystyle\int_{Sp(A)}
     \bra{\lambda_1}_{\mathcal{H}_1}\otimes
     \bra{\lambda_2}_{\mathcal{H}_2} 
    (\varphi_1\otimes\varphi_2)
     \ket{\lambda_1}_{\mathcal{H}_1}\otimes
     \ket{\lambda_2}_{\mathcal{H}_2}
    d\nu\\
    & =
    \displaystyle\int_{Sp(A)}
     \braket{\lambda_1}{\varphi_1}_{\mathcal{H}_1}
     \braket{\lambda_2}{\varphi_2}_{\mathcal{H}_2}
     \ket{\lambda_1}_{\mathcal{H}_1}\otimes
     \ket{\lambda_2}_{\mathcal{H}_2}
    d\nu.
    \end{split}
    \label{specialcase_spectralexpansion_ket_a}  
\end{align}
(\ref{specialcase_spectralexpansion_ket_a}) shows that the expansion coefficient of $\ket{\varphi_1}_{\mathcal{H}_1}\otimes\ket{\varphi_2}_{\mathcal{H}_2}$ by the set of the (generalized) eigenvectors $\{\ket{\lambda_1}_{\mathcal{H}_1}\otimes
     \ket{\lambda_2}_{\mathcal{H}_2}\}$ of $A$ is given as $\braket{\lambda_1}{\varphi_1}_{\mathcal{H}_1}
     \braket{\lambda_2}{\varphi_2}_{\mathcal{H}_2}$, where $\lambda=\lambda_1+\lambda_2$ goes through $Sp(A)$.
The expansion for the bra, $\bra{\varphi_1}_{\mathcal{H}_1}\otimes\bra{\varphi_2}_{\mathcal{H}_2}$, is also obtained as 
\begin{align}
\begin{split}
    \bra{\varphi_1}_{\mathcal{H}_1}\otimes\bra{\varphi_2}_{\mathcal{H}_2}
    & =
    \displaystyle\int_{Sp(A)}
     \braket{\varphi_1}{\lambda_1}_{\mathcal{H}_1}
     \braket{\varphi_2}{\lambda_2}_{\mathcal{H}_2}
     \bra{\lambda_1}_{\mathcal{H}_1}\otimes
     \bra{\lambda_2}_{\mathcal{H}_2}
    d\nu,
    \end{split}
    \label{specialcase_spectralexpansion_bra_a}  
\end{align}
whose expansion coefficient is $\braket{\varphi_1}{\lambda_1}_{\mathcal{H}_1}
     \braket{\varphi_2}{\lambda_2}_{\mathcal{H}_2}$.
%
%

%
%


\subsection{Eigenequations and Complete orthonormal system}
\label{sec:3.3}

The tensor products of generalized eigenvectors $\bra{\lambda_1}_{\mathcal{H}_1}\otimes
     \bra{\lambda_2}_{\mathcal{H}_2}$ and $\ket{\lambda_1}_{\mathcal{H}_1}\otimes
     \ket{\lambda_2}_{\mathcal{H}_2}$ for $A$ 
 satisfy the following eigenequations, respectively : for any $\varphi \in \Phi_1 \widehat{\otimes} \Phi_2$, 
\begin{eqnarray}
    \bra{\lambda_1}_{\mathcal{H}_1}\otimes
     \bra{\lambda_2}_{\mathcal{H}_2}(A\varphi)
     & = &
     (\lambda_1+\lambda_2)\bra{\lambda_1}_{\mathcal{H}_1}\otimes
     \bra{\lambda_2}_{\mathcal{H}_2}(\varphi),
     \label{eqn:eigenequation_tensors1}\\
     \ket{\lambda_1}_{\mathcal{H}_1}\otimes
     \ket{\lambda_2}_{\mathcal{H}_2}(A\varphi)
     & = &
     (\lambda_1+\lambda_2)\ket{\lambda_1}_{\mathcal{H}_1}\otimes
     \ket{\lambda_2}_{\mathcal{H}_2}(\varphi).
    \label{eqn:eigenequation_tensors2}
\end{eqnarray}
To see them, let $\varphi \in \Phi_1 \widehat{\otimes} \Phi_2$ 
and then $\varphi$ can be represented as the form of the sum of an absolutely convergence series, $\varphi =\sum_{i=1}^{\infty}r_i \varphi_i^1\otimes \varphi_i^2$,
where $\sum_{i}|r_i|\leq 1$ and $\{\varphi_i^1\}$ and $\{\varphi_i^2\}$ are null sequences in $\Phi_1$ and $\Phi_2$, respectively~\cite{Schaefer1966}.
As $A=\overline{A_1\otimes I_2+I_1\otimes A_2}$ is continuous linear on $\Phi_1 \widehat{\otimes} \Phi_2$, 
we obtain $A\varphi=\sum_{i=1}^\infty r_i (A_1\varphi_i^1\otimes\varphi_i^2+\varphi_i^1\otimes A_2\varphi_i^2)$.
Therefore, the continuous anti-linearity of  $\ket{\lambda_1}_{\mathcal{H}_1}\otimes
     \ket{\lambda_2}_{\mathcal{H}_2}$ on $\Phi_1 \widehat{\otimes} \Phi_2$ provides
\begin{eqnarray}
     \ket{\lambda_1}_{\mathcal{H}_1}\otimes
     \ket{\lambda_2}_{\mathcal{H}_2}(A\varphi)
     & = &
    \sum_{i=1}^\infty r^*_i \ket{\lambda_1}_{\mathcal{H}_1}\otimes
     \ket{\lambda_2}_{\mathcal{H}_2}
     (A_1\varphi_i^1\otimes\varphi_i^2+\varphi_i^1\otimes A_2\varphi_i^2)\nonumber\\
     & = &
     \sum_{i=1}^\infty r^*_i 
     \Big{(}\lambda_1\ket{\lambda_1}_{\mathcal{H}_1}\varphi_i^1
     \otimes
     \ket{\lambda_2}_{\mathcal{H}_2}\varphi_i^2
     +\lambda_2\ket{\lambda_1}_{\mathcal{H}_1}\varphi_i^1
     \otimes 
     \ket{\lambda_2}_{\mathcal{H}_2}\varphi_i^2\Big{)}\nonumber \\
     & = &
     (\lambda_1+\lambda_2)\ket{\lambda_1}_{\mathcal{H}_1}\otimes \ket{\lambda_2}_{\mathcal{H}_2}(\sum_{i=1}^{\infty}r_i \varphi_i^1\otimes \varphi_i^2)\nonumber \\
     & = &
     (\lambda_1+\lambda_2)\ket{\lambda_1}_{\mathcal{H}_1}\otimes \ket{\lambda_2}_{\mathcal{H}_2}(\varphi).
    \label{eqn:app.eigenequation_tensors}
\end{eqnarray}
Similarly, we have 
$\bra{\lambda_1}_{\mathcal{H}_1}\otimes
     \bra{\lambda_2}_{\mathcal{H}_2}(A\varphi)
     =     (\lambda_1+\lambda_2)\bra{\lambda_1}_{\mathcal{H}_1}\otimes \bra{\lambda_2}_{\mathcal{H}_2}(\varphi)$,
which shows the eigenequation (\ref{eqn:eigenequation_tensors2})

The complete orthonormal form is established using $\{\ket{\lambda_1}_{\mathcal{H}_1}\otimes
     \ket{\lambda_2}_{\mathcal{H}_2}\}$.
Actually, by using the expansions (\ref{spectralexpansion_ket_a}) or  (\ref{spectralexpansion_bra_a}), 
the completion form is given as 
\BA
    I & = &  
    \int_{\lambda\in Sp(A)}
    \int_{\lambda=\lambda_1+\lambda_2}
    \ket{\lambda_1}_{\mathcal{H}_1}\otimes
     \ket{\lambda_2}_{\mathcal{H}_2}\bra{\lambda_1}_{\mathcal{H}_1}\otimes
     \bra{\lambda_2}_{\mathcal{H}_2}
     d\sigma^\lambda_{\lambda_1,\lambda_2}d\mu_\lambda
    \nonumber\\
    & = & \int_{Sp(A)} 
    \ket{\lambda_1}_{\mathcal{H}_1}\otimes
     \ket{\lambda_2}_{\mathcal{H}_2}\bra{\lambda_1}_{\mathcal{H}_1}\otimes
     \bra{\lambda_2}_{\mathcal{H}_2}
     d\nu.
    \label{eqn:O3-3-4}
\EA
Here, the notation (\ref{abbreviation}) is adapted.
To consider the orthonormality, putting $\varphi(\lambda_1,\lambda_2)
\equiv
\bra{\lambda_1}_{\mathcal{H}_1}\otimes     \bra{\lambda_2}_{\mathcal{H}_2}(\varphi)=\bra{\lambda_1}_{\mathcal{H}_1}\otimes     \bra{\lambda_2}_{\mathcal{H}_2}\ket{\varphi}_{\mathcal{H}_1\overline{\otimes} \mathcal{H}_2}$ for $\varphi \in \Phi_1 \widehat{\otimes} \Phi_2$,
we have 
\BA
    &&
    \int_{\lambda\in Sp(A)}
    \int_{\lambda=\lambda_1+\lambda_2}
    \bra{\lambda^{\prime}_1}_{\mathcal{H}_1}\otimes     \bra{\lambda^{\prime}_2}_{\mathcal{H}_2}
    \ket{\lambda_1}_{\mathcal{H}_1}\otimes     \ket{\lambda_2}_{\mathcal{H}_2}
    \varphi(\lambda_1,\lambda_2) d\sigma^{\lambda}_{\lambda_1,\lambda_2}d\mu_{\lambda} \nonumber\\
    & = &
    \int_{\lambda\in Sp(A)}
    \int_{\lambda=\lambda_1+\lambda_2} 
    \bra{\lambda^{\prime}_1}_{\mathcal{H}_1}\otimes     \bra{\lambda^{\prime}_2}_{\mathcal{H}_2}
    \ket{\lambda_1}_{\mathcal{H}_1}\otimes     \ket{\lambda_2}_{\mathcal{H}_2}
    \bra{\lambda_1}_{\mathcal{H}_1}\otimes     \bra{\lambda_2}_{\mathcal{H}_2}
    \ket{\varphi}_{\mathcal{H}_1\overline{\otimes} \mathcal{H}_2}
    d\sigma^{\lambda}_{\lambda_1,\lambda_2}d\mu_{\lambda} \nonumber\\
    & = &
    \bra{\lambda^{\prime}_1}_{\mathcal{H}_1}\otimes     \bra{\lambda^{\prime}_2}_{\mathcal{H}_2}
        \ket{\varphi}_{\mathcal{H}_1\overline{\otimes} \mathcal{H}_2}
        \nonumber\\
    & = &
    \varphi(\lambda_1^\prime,\lambda_2^\prime).
    \label{eqn:O3-3-5}
\EA
(\ref{eqn:O3-3-5}) implies that the combination,
$\bra{\lambda^{\prime}_1}_{\mathcal{H}_1}\otimes     \bra{\lambda^{\prime}_2}_{\mathcal{H}_2}
    \ket{\lambda_1}_{\mathcal{H}_1}\otimes     \ket{\lambda_2}_{\mathcal{H}_2}$, of $\bra{\lambda^{\prime}_1}_{\mathcal{H}_1}\otimes     \bra{\lambda^{\prime}_2}_{\mathcal{H}_2}$
    and 
    $
    \ket{\lambda_1}_{\mathcal{H}_1}\otimes     \ket{\lambda_2}_{\mathcal{H}_2}$
can be represented by the product of $\delta$-functions,
\begin{eqnarray}
\bra{\lambda^{\prime}_1}_{\mathcal{H}_1}\otimes     \bra{\lambda^{\prime}_2}_{\mathcal{H}_2}
    \ket{\lambda_1}_{\mathcal{H}_1}\otimes     \ket{\lambda_2}_{\mathcal{H}_2}
    =
    \Check{\delta}(\lambda_1^\prime-\lambda_1)
    \Check{\delta}(\lambda_2^\prime-\lambda_2),
    \label{eqn:O3-3-6}
\end{eqnarray}
where $\Check{\delta}$ is performed as
\begin{eqnarray}
    f(\lambda_1^\prime,\lambda_2^\prime)
    & = & 
    \int_{\lambda\in Sp(A)}
    \int_{\lambda=\lambda_1+\lambda_2}
    f(\lambda_1,\lambda_2)\Check{\delta}(\lambda_1^\prime-\lambda_1)
    \Check{\delta}(\lambda_2^\prime-\lambda_2)
    d\sigma^{\lambda}_{\lambda_1,\lambda_2}d\mu_{\lambda}\nonumber \\
   & = &
    \int_{Sp(A)}
    f(\lambda_1,\lambda_2)\Check{\delta}(\lambda_1^\prime-\lambda_1)
    \Check{\delta}(\lambda_2^\prime-\lambda_2)
    d\nu
    \nonumber
    \label{deltafunction}
\end{eqnarray}
for any function $f(\lambda_1,\lambda_2)$.
Thus, the complete orthonormal form given by
$\{\ket{\lambda_1}_{\mathcal{H}_1}\otimes     \ket{\lambda_2}_{\mathcal{H}_2}\}$ is constructed as the relations (\ref{eqn:O3-3-4}) and (\ref{eqn:O3-3-6}).
%


\subsection{Extension to the dual spaces}
\label{sec:3.4}

The self-adjoint operator $
A=\overline{A_1\otimes I_2 +I_1\otimes A_2} 
$ can be extended to dual spaces as follows.
As $A$ is continuous on $(\Phi_1 \hat{\otimes} {\Phi}_2,\widehat{\tau_p})$ with $A(\Phi_1 \hat{\otimes} {\Phi}_2) \subset \Phi_1 \hat{\otimes} {\Phi}_2$, an operator
\begin{eqnarray}
    \hat{A} : (\Phi_1 \widehat{\otimes} \Phi_2)^\prime \cup (\Phi_1 \widehat{\otimes} \Phi_2)^\times
    \rightarrow
    (\Phi_1 \widehat{\otimes} \Phi_2)^\prime \cup (\Phi_1 \widehat{\otimes} \Phi_2)^\times
    \label{extension_of_A}
\end{eqnarray}
can be defined as
\begin{eqnarray}
   (\hat {A} (f))(\varphi):=f(A(\varphi)),
    \label{eqn:O3-4-1}
\end{eqnarray}
for any $\varphi\in \Phi_1 \widehat{\otimes} \Phi_2$ and $f\in(\Phi_1 \widehat{\otimes} \Phi_2)^\prime \cup (\Phi_1 \widehat{\otimes} \Phi_2)^\times$.
It follows from (\ref{eqn:eigenequation_tensors1}) and (\ref{eqn:eigenequation_tensors2}) that $\hat {A}$ satisfies the eigenequations 
 with respect to $\{\bra{\lambda_1}_{\mathcal{H}_1}\otimes
     \bra{\lambda_2}_{\mathcal{H}_2}\}$ and $\{\ket{\lambda_1}_{\mathcal{H}_1}\otimes
     \ket{\lambda_2}_{\mathcal{H}_2}\}$,
\begin{eqnarray}
    \bra{\lambda_1}_{\mathcal{H}_1}\otimes
     \bra{\lambda_2}_{\mathcal{H}_2}\hat{A}
     & = &
     (\lambda_1+\lambda_2)\bra{\lambda_1}_{\mathcal{H}_1}\otimes
     \bra{\lambda_2}_{\mathcal{H}_2},
     \label{eqn:O3-4-3a}\\
     \hat{A}\ket{\lambda_1}_{\mathcal{H}_1}\otimes
     \ket{\lambda_2}_{\mathcal{H}_2}
     & = &
     (\lambda_1+\lambda_2)\ket{\lambda_1}_{\mathcal{H}_1}\otimes
     \ket{\lambda_2}_{\mathcal{H}_2}.
    \label{eqn:O3-4-3b}
\end{eqnarray}
%
%

%
In each RHS, $\Phi_i \subset \mathcal{H}_i \subset \Phi _i^\prime, \Phi _i^\times$, the self-adjoint operator $A_i : D(A_i)\rightarrow \mathcal{H}_i$ is assumed to be continuous on $\Phi_i$ and $A_i(\Phi_i)\subset \Phi_i$ ($i=1,2$).
Therefore, there corresponds the extension $\hat{A_i}$ on $\Phi_i^\prime\cup\Phi_i^\times$ to each $A_i$ such that
$(\hat {A_i} (f))(\varphi):=f(A_i(\varphi))$
for any $\varphi\in \Phi_i$ and $f\in\Phi_i ^\prime \cup \Phi_i^\times$.

Between the extensions $\hat{A}_i$ $(i=1,2)$ and $\hat{A}$ defined in (\ref{eqn:O3-4-1}), 
the relation 
\begin{eqnarray}
 \hat{A}=\hat{A_1}\otimes\hat{I_2}+\hat{I_1}\otimes\hat{A_2}
    \label{eqn:extension_relation}
\end{eqnarray}
holds on the subset $(\Phi_1^\prime\otimes\Phi_2^\prime)\cup(\Phi_1^\times\otimes\Phi_2^\times)$ of $(\Phi_1 \widehat{\otimes} \Phi_2)^\prime \cup (\Phi_1 \widehat{\otimes} \Phi_2)^\times$, 
where $\hat{I_i}$ is the identity on $\Phi_i^\prime\cup\Phi_i^\times$.
(For the proof, see Appendix \ref{sec:app.extension}).
This relation shows the connection of the self-adjoint observable $\hat{A_i} (i=1,2)$ of the isolated systems with $\hat{A}$ of their composite system in the dual spaces.   

Notably, the relations obtained so far can be easily generalized to the $N$-tensor product of RHS (\ref{eqn:O2-2-1}) using the self-adjoint operator $A=\overline{\sum_{i=1}^N\Check{A_i}}$ where $\Check{A_i}=I\otimes I\otimes \dots \otimes I \otimes A_i \otimes I \otimes\dots \otimes I$.


\section{Symmetry of the bra-ket in the dual spaces}
\label{sec:4}

In Sec.\ref{sec:2.2}, we have constructed the extension $\widetilde{P_c^{\widehat{\otimes}^N {\Phi}}}$ of the projection related to the permutation operator to the dual spaces, defined by (\ref{eqn:o3-2-1}),  such that the symmetry of the bra-ket vectors defined in the dual spaces is provided by $\widetilde{P_c^{\widehat{\otimes}^N {\Phi}}}$.
This section focuses on the relationship between $\widetilde{P_c^{\widehat{\otimes}^N {\Phi}}}$ and the generalized eigenvectors obtained via the nuclear spectral theorem for $A$ of the form (\ref{composed_operator}). 
It also ensures the symmetry of the bra-ket notation in the dual spaces based on the RHS formulation.

Now we assume $\mathcal{H}_1=\mathcal{H}_2=\dots=\mathcal{H}_N\equiv \mathcal{H}$ 
and
$\Phi_1=\Phi_2=\dots=\Phi_N \equiv \Phi$.
We set the RHS comprising the $N$-tensor product of RHS (\ref{eqn:O3-1})
and consider a self-adjoint operator in the RHS,  
\begin{eqnarray}
    A=\overline{\sum_{i=1}^N\Check{A_i}} \text{~~~where~~~} \Check{A_i}=I\otimes I\otimes \dots \otimes I \otimes A_i \otimes I \otimes\dots \otimes I.
    \label{operator_of_A_identity}
\end{eqnarray}
Its extension given by (\ref{eqn:O3-4-1}) to $(\widehat{\otimes}^N \Phi)^{\prime}\cup
(\widehat{\otimes}^N \Phi)^{\times}$ is denoted by $\hat{A}$.
Subsequently, $A$ has the complete orthonormal system composed of the eigenvectors $\{\ket{\lambda_1}_{\mathcal{H}}\otimes \dots\otimes \ket{\lambda_N}_{\mathcal{H}}\}$. 
Here,  $\bra{\lambda_1}_{\mathcal{H}}\otimes \dots\otimes \bra{\lambda_N}_{\mathcal{H}}$ and $\ket{\lambda_1}_{\mathcal{H}}\otimes \dots\otimes \ket{\lambda_N}_{\mathcal{H}}$ belong to the dual spaces $(\widehat{\otimes}^N \Phi)^{\prime}$ and $
(\widehat{\otimes}^N \Phi)^{\times}$ and satisfy the 
eigenequations (\ref{eqn:O3-4-3a}) and (\ref{eqn:O3-4-3b}), respectively.
The symmetry of $\ket{\varphi_1,\otimes\dots \otimes \varphi_N}_{\overline{\otimes}^N\mathcal{H}}=\ket{\varphi_1}_{\mathcal{H}}\otimes \dots\otimes \ket{\varphi_N}_{\mathcal{H}}$ is determined as the relation (\ref{eqn:o3-ket-symmetry}) using $\widetilde{P_c^{\widehat{\otimes}^2 {\Phi}}}$.
We now consider the transformation of $\ket{\lambda_1}_{\mathcal{H}}\otimes \dots\otimes \ket{\lambda_N}_{\mathcal{H}}$ by $\widetilde{P_c^{\widehat{\otimes}^2 {\Phi}}}$.
For any $\phi=\phi_1\otimes\dots \otimes\phi \in {\otimes}^N \Phi_N \subset \widehat{\otimes}^N \Phi$, 
we have 
\begin{eqnarray}
    \widetilde{P_c^{\widehat{\otimes}^2 {\Phi}}}(\ket{\lambda_1}_{\mathcal{H}}\otimes \dots\otimes \ket{\lambda_N}_{\mathcal{H}})
    (\phi)
    & = &
    \ket{\lambda_1}_{\mathcal{H}}\otimes \dots\otimes \ket{\lambda_N}_{\mathcal{H}}(P_{c_1}^{\widehat{\otimes}^2 {\Phi}}(\phi))
    \nonumber \\
    & = &
    \ket{\lambda_1}_{\mathcal{H}}\otimes \dots\otimes \ket{\lambda_N}_{\mathcal{H}}
    \Bigr\{ \displaystyle\frac{1}{N!}\sum_{\sigma \in \mathfrak{S}_n}c(\sigma)\ket{\phi_{\sigma(1)}}_{\mathcal{H}}
    \otimes \dots \otimes
    \ket{\phi_{\sigma(N)}}_{\mathcal{H}}
    \Bigr\}
     \nonumber\\
    & = &
    \displaystyle\frac{1}{N!}\sum_{\sigma \in \mathfrak{S}_n}c(\sigma)\ket{\lambda_1}_{\mathcal{H}}\otimes \dots\otimes \ket{\lambda_N}_{\mathcal{H}}
    \ket{\phi_{\sigma(1)}}_{\mathcal{H}}
    \otimes \dots \otimes
    \ket{\phi_{\sigma(N)}}_{\mathcal{H}}
     \nonumber\\
    & = &
    \displaystyle\frac{1}{N!}\sum_{\sigma \in \mathfrak{S}_n}c(\sigma)
    \braket{\phi_{\sigma(1)}}{\lambda_1}
    \dots
    \braket{\phi_{\sigma(N)}}{\lambda_N}
    \nonumber\\
    & = &
    \displaystyle\frac{1}{N!}\sum_{\sigma \in \mathfrak{S}_n}c(\sigma)
    \braket{\phi_{1}}{\lambda_{\sigma(1)}}
    \dots
    \braket{\phi_{N}}{\lambda_{\sigma(N)}}
    \nonumber \\
    & = &
    \displaystyle\frac{1}{N!}\sum_{\sigma \in \mathfrak{S}_n}c(\sigma)\ket{\lambda_{\sigma(1)}}_{\mathcal{H}}
    \otimes \dots \otimes
    \ket{\lambda_{\sigma(N)}}_{\mathcal{H}}(\phi).
    \label{eqn:O3-5-calculation1}
\end{eqnarray}
From (\ref{eqn:O3-5-calculation1}), it is confirmed that 
\begin{eqnarray}
    \widetilde{P_c^{\widehat{\otimes}^2 {\Phi}}}\ket{\lambda_1}_{\mathcal{H}}\otimes \dots\otimes \ket{\lambda_N}_{\mathcal{H}}
    =
    \displaystyle\frac{1}{N!}\sum_{\sigma \in \mathfrak{S}_n}c(\sigma)\ket{\lambda_{\sigma(1)}}_{\mathcal{H}}
    \otimes \dots \otimes
    \ket{\lambda_{\sigma(N)}}_{\mathcal{H}}.  
    \label{eqn:O3-5-1}
\end{eqnarray}
(\ref{eqn:O3-5-1}) shows that the permutation operator $\widetilde{P_c^{\widehat{\otimes}^2 {\Phi}}}$ determines the symmetric structure of the eigenvectors $\ket{\lambda_1}_{\mathcal{H}}\otimes \dots\otimes \ket{\lambda_N}_{\mathcal{H}}$.
Similarly, 
we obtain  
\begin{eqnarray}
    \bra{\lambda_1}_{\mathcal{H}}\otimes \dots\otimes \bra{\lambda_N}_{\mathcal{H}}\widetilde{P_c^{\widehat{\otimes}^2 {\Phi}}}
    =
    \displaystyle\frac{1}{N!}\sum_{\sigma \in \mathfrak{S}_n}c(\sigma)\bra{\lambda_{\sigma(1)}}_{\mathcal{H}}
    \otimes \dots \otimes
    \bra{\lambda_{\sigma(N)}}_{\mathcal{H}}.  
    \label{eqn:O3-5-2}
\end{eqnarray}
Therefore, based on these relations (\ref{eqn:O3-5-1}) and (\ref{eqn:O3-5-2}), in addition to the results shown in Sec.\ref{sec:2.2}, one can conclude that the operator $\widetilde{P_c^{\widehat{\otimes}^2 {\Phi}}}$  completely provides the symmetry of the bra-ket vectors constructed in the RHS formulation.

In quantum mechanics, the commutative relation between an observable and the permutation operator is considered as the fundamental condition for proving that the symmetric and anti-symmetric states of identical particles become the eigenvectors of the observable~\cite{Messiah}.  
According to the proposed RHS framework, 
the operators (\ref{eqn:O3-5-1}) and (\ref{eqn:O3-5-2}) become the generalized eigenvectors of $A$ when $A$ and  $P_c^{\widehat{\otimes}^2 {\Phi}}$ are commutative on $\widehat{\otimes}^2 {\Phi}$, namely,  
\begin{eqnarray}
    \text{$[A, P_c^{\widehat{\otimes}^2 {\Phi}}] =0$ ~on $\widehat{\otimes}^2 {\Phi}$.}
    \label{eqn:O3-5-condition}
\end{eqnarray}
%
%
To verify this fact, the following lemma is applicable.
\begin{lemma}
\label{lemma}
{
Let $\Phi \subset \mathcal{H} \subset \Phi^\prime, \Phi^\times$ be an RHS and 
let $A:D(A)\to \mathcal{H}$ and $B:D(B)\to \mathcal{H}$ be self-adjoint operators in $\mathcal{H}$ such that they are continuous on $\Phi$ and the $A\Phi \subset \Phi$ and $B\Phi \subset \Phi$ are satisfied.
If $A$ and $B$ commute on $\Phi$, for each generalized eigenbra $\bra{\lambda}$ and eigenket $\ket{\lambda}$ of $A$ corresponding to the eigen value $\lambda$, 
the elements $\bra{\lambda}\hat{B}\in \Phi^\prime$ and $\hat{B}\ket{\lambda}\in \Phi^\times$ are the generalized eigen bra and ket, corresponding to $\lambda$.
Consequently, the relations 
\begin{eqnarray}
    \bra{\lambda}\hat{B} \hat{A}=\lambda\bra{\lambda}\hat{B},~~
    \hat{A}\hat{B}\ket{\lambda}=\lambda\hat{B}\ket{\lambda}
    \label{eqn:O3-5-Prop}
\end{eqnarray}
are satisfied, where $\hat{A}$ and $\hat{B}$ are the extensions on $\Phi^\prime \cup \Phi^\times$. 
}
\end{lemma}
\begin{proof}
    As $\ket{\lambda}$ is a generalized eigenvector of $A$ corresponding to $\lambda$, 
$\hat{A}\ket{\lambda}(\varphi)=\lambda\ket{\lambda}(\varphi)$ for any $\varphi\in \Phi$ is satisfied.
Let $\varphi\in \Phi$. 
Noting the anti-linearity of $\ket{\lambda}$ and $\hat{B}\ket{\lambda}$,
we have,
\begin{eqnarray}
    \hat{A}(\hat{B}\ket{\lambda})(\varphi)
    & = &
    \hat{B}\ket{\lambda}(A\varphi)
    =
    \ket{\lambda}(B(A\varphi))
    =
    \ket{\lambda}(A(B\varphi))\nonumber\\
    & = &
    \lambda\ket{\lambda}(B\varphi)=\ket{\lambda}(B\lambda^*\varphi)=\hat{B}(\ket{\lambda})(\lambda^*\varphi)=\lambda\hat{B}\ket{\lambda}(\varphi).
    \label{eqn:app.3.5}
\end{eqnarray}
Similarly, $\hat{A}(\hat{B}\bra{\lambda})(\varphi)=\lambda\bra{\lambda}\hat{B}(\varphi)$.
Thus, the desired assertion is complete.  

\end{proof}

Based on this lemma, 
it can be easily shown that when the condition (\ref{eqn:O3-5-condition}) is satisfied, (\ref{eqn:O3-5-1}) and (\ref{eqn:O3-5-2}) are the generalized eigenvectors of $A$ that satisfy
\begin{eqnarray}
    \hat{A}\widetilde{P_c^{\widehat{\otimes}^2 {\Phi}}}(\ket{\lambda_1}_{\mathcal{H}}\otimes \dots\otimes \ket{\lambda_N}_{\mathcal{H}})
    =
    (\lambda_1+ \dots+ \lambda_N)
    \widetilde{P_c^{\widehat{\otimes}^2 {\Phi}}}(\ket{\lambda_1}_{\mathcal{H}}\otimes \dots\otimes \ket{\lambda_N}_{\mathcal{H}}),
    \label{eqn:O3-5-3a}
\end{eqnarray}
and 
\begin{eqnarray}
    \bra{\lambda_1}_{\mathcal{H}}\otimes \dots\otimes \bra{\lambda_N}_{\mathcal{H}}\widetilde{P_c^{\widehat{\otimes}^2 {\Phi}}}\hat{A}
    =
    (\lambda_1+ \dots+ \lambda_N)
    \bra{\lambda_1}_{\mathcal{H}}\otimes \dots\otimes \bra{\lambda_N}_{\mathcal{H}}\widetilde{P_c^{\widehat{\otimes}^2 {\Phi}}}.
    \label{eqn:O3-5-3b}
\end{eqnarray}




%
%
%
\color{black}



\section{Conclusion}
\label{sec:5}

This study discussed the mathematical treatment of Dirac’s bra-ket formalism for composite systems in addition to identical particle systems using the RHS approach.
The tensor product of an RHS facilitates the precise construction of bra and ket vectors in the dual spaces. 
For identical particles systems, 
the symmetric structure of the bra and ket vectors can be introduced by extending the permutation operator to the dual spaces.
The spectral expansions of bra and ket vectors for a self-adjoint operator corresponding to an observable in composite systems via its generalized eigenvectors was established.
%
These generalized eigenvectors were associated with the eigenvectors of self-adjoint operators for single particles
and established the complete orthonormal system.
Furthermore, we investigate a relationship between the generalized eigenvectors for the self-adjoint operator and the extended permutation operator for preserving the symmetric structure of the bra and ket vectors in the dual spaces.
In future work, the present RHS formulation will be applied to areas such as quantum statistical mechanics and quantum field theory, aiming to enable more precise discussions of established studies~\cite{Antoiou1998, Antoiou2003, Liu2013}.

\appendix

\numberwithin{equation}{section}
\makeatletter

\section{The relation (\ref{eqn:extension_relation})}
\label{sec:app.extension}

Let $f\in \Phi_1^\prime \otimes \Phi_2^\prime$.
There are sequences $\{f_i^1\}$ and $\{f_i^2\}$ in $\Phi_1^\prime$ and $\Phi_2^\prime$ such that $f=\sum_{i=1}^n f_i^1\otimes f_i^2$.
By $(\hat{A_1}\otimes\hat{I_2}+\hat{I_1}\otimes \hat{A_2})(f)=\sum_{i=1}^n(\hat{A_1}f_i^1\otimes f_i^2)+\sum_{i=1}^n(f_i^1\otimes \hat{A_2}f_i^2)$, for $\phi=\sum_{j=1}^m\phi_j^1\otimes \phi_j^2\in \Phi_1\otimes \Phi_2$, we have
\begin{eqnarray}
    (\hat{A_1}\otimes\hat{I_2}+\hat{I_1}\otimes \hat{A_2})(f)(\phi) & = & \Big(\sum_{i=1}^n(\hat{A_1}f_i^1\otimes f_i^2)+\sum_{i=1}^n(f_i^1\otimes \hat{A_2}f_i^2)\Big)(\phi)
    \nonumber \\
    & = &
    \sum_{i=1}^n\sum_{j=1}^m\Big\{(\hat{A_1}f_i^1)(\phi_j^1)\otimes f_i^2(\phi_j^2)+f_i^1(\phi_j^1)\otimes(\hat{A_2}f_i^2)(\phi_j^2)\Big\}
    \nonumber \\
    & = &
    \sum_{i=1}^n\sum_{j=1}^m\Big\{f_i^1(A_1\phi_j^1)\otimes f_i^2(\phi_j^2)+f_i^1(\phi_j^1)\otimes f_i^2(A_2\phi_j^2)\Big\}.
    \label{eqn:appendix_1}
\end{eqnarray}
On the other hand,
since $\Phi_1^\prime \otimes \Phi_2^\prime$ is a subspace of $(\Phi_1\otimes\Phi_2)^\prime$ and $\hat{A}$ is linear on $(\Phi_1\otimes\Phi_2)^\prime$, 
\begin{eqnarray}
    \hat{A}(f)(\phi) & = & \sum_{i=1}^n \hat{A}(f_i^1\otimes f_i^2)(\phi)
    =\sum_{i=1}^n(f_i^1\otimes f_i^2)(A\phi)
    \nonumber \\
    & = &
    \sum_{i=1}^nf_i^1\otimes f_i^2\Big\{ \sum_{j=1}^m A_1\phi_j^1\otimes \phi_j^2+\phi_j^1 \otimes A_2\phi_j^2 \Big\}
    \nonumber \\
    & = &
    \sum_{i=1}^n\sum_{j=1}^m\Big\{f_i^1(A_1\phi_j^1)\otimes f_i^2(\phi_j^2)+f_i^1(\phi_j^1)\otimes f_i^2(A_2\phi_j^2)\Big\}.
    \label{eqn:appendix_2}
\end{eqnarray}
Thus, for any $f\in \Phi_1^\prime\otimes \Phi_2^\prime$, $\hat{A}(f)=(\hat{A_1}\otimes \hat{I_2}+\hat{I_1}\otimes \hat{A_2})(f)$.
Similarly, we can show $\hat{A}(f)=(\hat{A_1}\otimes \hat{I_2}+\hat{I_1}\otimes \hat{A_2})(f)$ for any $f\in \Phi_1^\times \otimes \Phi_2^\times$.
Thus, the relation (\ref{eqn:extension_relation}) holds on $(\Phi_1^\prime \otimes \Phi_2^\prime)\cup(\Phi_1^\times \otimes \Phi_2^\times)$.
%

\bigskip

\noindent
{\bf Acknowledgement}

The authors are grateful to
Prof. Y.~Yamazaki, Prof. T.~Yamamoto, Prof. Y.~Yamanaka, Prof. K.~Iida,
Prof. H.~Ujino, Prof. I.~Sasaki, 
Prof. H.~Saigo, Prof. F.~Hiroshima, 
Prof. S. Matsutani,
and Emeritus A.~Kitada for their useful comments and encouragement.
This work was supported by the Sasakawa Scientific Research Grant from The Japan Science Society and JSPS KAKENHI Grant Number 22K13976.



\bigskip


\end{document}